\documentclass{article}
\usepackage{graphicx}
\usepackage{subfigure}
\usepackage{amsfonts}
\usepackage{amsmath}
\usepackage{amssymb}
\usepackage{url}
\usepackage{fancyhdr}
\usepackage{indentfirst}
\usepackage{enumerate}
\usepackage{dsfont}
  \usepackage[colorlinks=true,citecolor=blue]{hyperref}
\usepackage{amsthm}
\usepackage{color}
\usepackage{natbib}
\usepackage{comment}
\usepackage{capt-of}

\def\d{\mathrm{d}}

\newcommand{\D}{\mathcal {D}}

\newcommand{\C}{\mathcal {C}}

\newcommand{\VaR}{\mathrm{VaR}}

\newcommand{\ES}{\mathrm{ES}}

\newcommand{\E}{\mathbb{E}}

\newcommand{\R}{\mathbb{R}}

\newcommand{\p}{\mathbb{P}}

\newcommand{\M}{\mathcal{M}}

\newcommand{\id}{\mathds{1}}
\newcommand{\lcx}{\prec_{\mathrm{cx}}}
\newcommand{\lst}{\prec_{\mathrm{st}}}

\renewcommand{\)}{\right)}

\renewcommand{\ge}{\geqslant}
\renewcommand{\le}{\leqslant}
\renewcommand{\geq}{\geqslant}
\renewcommand{\leq}{\leqslant}
\renewcommand{\epsilon}{\varepsilon}

\renewcommand{\cdots}{\dots}

\theoremstyle{plain}
\newtheorem{theorem}{Theorem}
\newtheorem{corollary}{Corollary}
\newtheorem{lemma}{Lemma}
\newtheorem{proposition}{Proposition}
\theoremstyle{definition}
\newtheorem{definition}{Definition}
\newtheorem{example}{Example}

\newtheorem{conjecture}{Conjecture}
\theoremstyle{remark}
\newtheorem{remark}{Remark}

\newcommand{\cet}{\begin{center}}
\newcommand{\ecet}{\end{center}}

\usepackage{setspace}

\begin{document}

\title{Ordering and Inequalities for Mixtures on Risk Aggregation}

\author{
  Yuyu Chen\thanks{\scriptsize Department of Statistics and Actuarial Science, University of Waterloo, Canada. Email: \texttt{y937chen@uwaterloo.ca}}
  \and Peng Liu\thanks{\scriptsize Department of Mathematical Sciences, University of Essex, UK. Email: \texttt{peng.liu@essex.ac.uk}}
  \and Yang Liu\thanks{\scriptsize Corresponding author. Department of Mathematical Sciences, Tsinghua University, China. Email: \texttt{yang-liu16@mails.tsinghua.edu.cn}}
  \and Ruodu Wang\thanks{\scriptsize Department of Statistics and Actuarial Science, University of Waterloo, Canada. Email: \texttt{wang@uwaterloo.ca}}
}

\maketitle


\begin{abstract}
	Aggregation sets, which represent  model uncertainty due to unknown dependence, are an important object in the study of robust risk aggregation. In this paper, we investigate  ordering relations  between two aggregation sets  for which the sets of marginals are related by two simple operations: distribution mixtures and quantile mixtures. Intuitively, these operations ``homogenize"   marginal distributions by making them similar. As a  general conclusion from our results, more ``homogeneous"  marginals lead  to a larger aggregation set, and thus more severe model uncertainty, although the situation for quantile mixtures is much more complicated than   that for distribution mixtures.
We proceed to study inequalities on the worst-case values of  risk measures in risk aggregation, which represent conservative calculation of regulatory  capital.  Among other results, we obtain an order relation on VaR under quantile mixture for marginal distributions with monotone densities.  Numerical results are presented to visualize the theoretical results and further inspire some conjectures. Finally, we provide applications on portfolio diversification under dependence uncertainty and merging p-values in multiple hypothesis testing, and discuss the connection of our results to joint mixability.

\end{abstract}


\newpage
\section{Introduction}

Robust risk aggregation has been studied extensively with applications in  banking  and insurance. A typical problem in this area is to compute the worst-case values of some risk measures for an aggregate loss with  unknown dependence structure. Two popular regulatory risk measures used in industry are Value-at-Risk (VaR) and the Expected Shortfall (ES); see \cite{MFE15}
and the references therein.
The worst-case value of ES in risk aggregation is explicit since ES is a  coherent risk measure (\cite{ADEH99}), whereas the worst-case value of VaR in risk aggregation generally does not admit analytical formulas, which is a known challenging problem (see e.g., \cite{EPR13, EWW15}).
See \cite{CLW18} on robust risk aggregation   for general risk measures, and  \cite{EKP20} on computation of robust risk aggregation using neural networks.

The above robust risk aggregation problem involves taking the supremum of a risk measure over an \textit{aggregation set}. Fix an atomless probability space $(\Omega,\mathcal{F},\p)$ and let $\mathcal M$ be the set of cdfs\footnote{In this paper, we treat probability measures on $\mathcal B(\R)$ and cdfs on $\R$ as equivalent objects.} on $\R$. For $F\in \mathcal M$,   $X\sim F$ means that the cdf of a random variable $X$ is $F$. Moreover, let $\mathcal M_1$ denote the set of cdfs on $\R$ with finite mean. For $\mathbf F=(F_1,\dots,F_n)\in \mathcal M^n$, the  {aggregation set} (\cite{BJW14}) is defined as
\begin{equation}\label{eq:aggregation_set}
\mathcal  D_n(\mathbf F)= \{\mbox{cdf of }X_1+\dots+X_n: X_i\sim F_i, ~i=1,\cdots,n\}.
 \end{equation}
The obvious interpretation is that $\mathcal D_n(\mathbf F)$ fully describes  model uncertainty associated with known marginal distributions $F_1,\dots,F_n$ but unknown dependence structure. The separate modeling of marginals and dependence is a standard practice in quantitative risk modeling, often involving copula techniques; see e.g., \cite{MFE15}.
An analytical characterization of $\D_n(\mathbf F)$ for a given $\mathbf F$ is very difficult and challenging. The only available analytical results are in \cite{MWW19} for standard uniform marginals.

The main objective of this paper  is to compare   model uncertainty of risk aggregation for $\mathbf F, \mathbf G \in \M^n$ which represent two possible  models of marginals.
The strongest form of comparison is set inclusion between two aggregation sets $\mathcal  D_n(\mathbf F )$ and $\mathcal  D_n(\mathbf G)$.
It turns out that such a strong relation may be achievable if  $\mathbf F, \mathbf G \in \mathcal M^n$ are related by the simple operations of distribution mixtures and quantile mixtures. Distribution mixture produces a tuple whose components are convex combinations of the given  distributions and quantile mixture yields a tuple whose components are given by   convex combinations of the given quantiles.
Both types of operations are common in statistics and risk management, as they correspond to simple operations on the parameters in statistical models or on portfolio construction; see Section \ref{sec:application} for an example. Moreover, if $\mathbf G$ is obtained from $\mathbf F$ via a distribution or quantile mixture, then the mean (assumed to be finite) of any element of $\mathcal  D_n(\mathbf G)$ is the same as that of any element  of $\mathcal  D_n(\mathbf F)$, making the comparison fair.
To the best of our knowledge, this paper is the first systematic study   on the order relation between $\D_n(\mathbf F)$ and $\D_n(\mathbf G)$ for different $\mathbf F$ and $\mathbf G$, thus comparing model uncertainty at the level of all possible distributions.

In some cases, a strong comparison via set inclusion is not possible, but we can compare values of a chosen risk measure.
For a law-invariant risk measure\footnote{We conveniently treat law-invariant risk measures as mappings on $\M$, although it is conventional to treat them as mappings on a space of random variables. The two settings are equivalent for law-invariant risk measures. } $\rho: \M \rightarrow \R$, we denote by $\overline{\rho}(\mathbf F)$ the worst-case value of $\rho$ in risk aggregation for $\mathbf F\in \mathcal M^n$, that is,
\begin{align*}
\overline{\rho}(\mathbf F)& = \sup \{\rho(F): F\in \mathcal D_n(\mathbf F)\}.
\end{align*}
We shall compare $ \overline{\rho}(\mathbf F)$ with $\overline{\rho}(\mathbf G)$, thus the worst-case values of a risk measure under model uncertainty, which usually represent conservative calculation of regulatory risk capital (e.g., \cite{EPR13}).
Certainly, $D_n(\mathbf F) \subset D_n(\mathbf G)$ implies $\overline{\rho}(\mathbf F) \leq \overline{\rho}(\mathbf G)$ for all risk measures $\rho$,
implying that the first comparison is stronger than the second one.\footnote{In this paper, the set inclusion ``$\subset$" is non-strict; the strict set inclusion is ``$\subsetneq$".
Similarly, the terms ``increasing" and ``decreasing" are in the non-strict sense.
}

Our study brings insights to two relevant problems in risk management. First, suppose that $\mathbf F$ and $\mathbf G$ are two possible statistical models for the marginal distributions in a risk aggregation setting. Our results allow for a comparison of model uncertainty associated with the two models, regardless of the choice of risk measures.
Although a completely unknown dependence structure is sometimes unrealistic, it is commonly agreed that the dependence structure in a risk model is difficult to accurately specify (e.g., \cite{EPR13} and \cite{BRV17}).
Hence, a comparison of the magnitude of model uncertainty is an important practical issue.
On the other hand,  the general conclusions remain valid even if the marginal distributions are not completely specific (see the discussion  in Section \ref{sec:conc} on the presence of marginal uncertainty), and thus the assumption of known marginal distributions  in our study is not harmful.

Second, our results provide an analytical way to establish   inequalities on the worst-case risk measures in the form $\overline{\rho}(\mathbf F)\le \overline{\rho}(\mathbf G)$. Sometimes the worst-case risk measure is difficult to calculate for $\mathbf F$, but it may be easier to calculate for $\mathbf G$. For instance,  formulas on worst-case VaR are available for some homogeneous marginal distributions in \cite{WPY13} and \cite{PR13},  but explicit results on heterogeneous marginal distributions are limited (see \cite{BLLW20} for a recent treatment). Therefore,  we can use the analytical formula $\overline{\rho}(\mathbf G)$, if available, as an upper bound on  $\overline{\rho}(\mathbf F)$, and this leads to interesting applications in other fields; see Section \ref{sec:application} for applications on portfolio diversification and multiple hypothesis testing and Section \ref{sec:jm} for a connection to joint mixability.

Our theoretical contributions are briefly summarized below.
In Sections \ref{sec:2} and \ref{sec:harmo}, we analyze general relations on distribution and quantile mixtures.
The general message of our results is that the more ``homogeneous" the distribution tuple is, the larger its corresponding aggregation set $\mathcal D_n$ is. 
In particular, the set inclusion   is  established for any tuples connected by distribution mixtures in Theorem \ref{th:1}; that is, $\mathcal D_n(\mathbf F) \subset \mathcal D_n(\mathbf G)$ if $\mathbf G$ is a distribution mixture of $\mathbf F$.
The problem for quantile mixtures is much more challenging. The set inclusion is established for uniform marginals in Proposition \ref{prop:uniform}. For other families of distributions, such a general relationship does not hold, as discussed with some examples.  

In Section \ref{sec:4}, we obtain inequalities between the worst-case values of some risk measure $\rho$ in risk aggregation with marginals  related by distribution or quantile mixtures.
Although quantile mixtures do not satisfy the relationship
$\mathcal D_n(\mathbf F) \subset \mathcal D_n(\mathbf G)$ in general,
we can prove an order property between $\overline{\rho}(\mathbf G)$ and $\overline{\rho}(\mathbf F)$ for commonly used risk measures.
Most remarkably, in Theorem \ref{th:scale}, we show that under a monotone density assumption, VaR satisfies this order property for a quantile mixture.
Section \ref{sec:bound} is dedicated to the most interesting special case of Pareto risk aggregation, with a special focus on the case of infinite mean.

Numerical results are presented in Section \ref{sec:numerical} to illustrate the  obtained results. In Section \ref{sec:application}, we provide two applications:   portfolio diversification under dependence uncertainty and merging p-values in multiple hypothesis testing. Some further technical discussions on distribution and quantile mixtures are put in Section \ref{sec:jm}. Section \ref{sec:conc} concludes the paper by presenting several open mathematical challenges related to quantile mixtures.  Some  proofs and further properties of Pareto risk aggregation are put in the   Appendix.

\section{Distribution mixtures}
\label{sec:2}


In this section we put our focus on one of the two operations: distribution mixture. The main objective is to establish some ordering relationships on the set $\mathcal D_n(\mathbf F)$ and $\mathcal D_n(\mathbf G)$ where $\mathbf G$ is a distribution mixture of $\mathbf F$. For greater generality, we investigate a more general $f$-aggregation set $\D_f(\mathbf F)$, where
$f: \R^n \to \R$ is a measurable and symmetric function.\footnote{A function $f$ is symmetric if $f(\mathbf x) = f(\pi (\mathbf x))$ for any $\mathbf  x \in \R^n$ and  $n$-permutation $\pi$.} Similarly to \eqref{eq:aggregation_set}, for $\textbf{F} = (F_1, \cdots, F_n) \in \M^n$, the $f$-aggregation set is defined as
$$
\D_f(\textbf{F}) = \{\text{cdf of } f\(X_1,\dots, X_n\): X_i \sim F_i, ~ i = 1, \cdots, n \}.
$$
It is clear that $\D_n$, defined in \eqref{eq:aggregation_set}, becomes a specific case of $D_f$ if $f$ is a sum function ($f(x_1, \cdots, x_n) = \sum_{j=1}^n x_j$). We first present some properties of the $f$-aggregation set.
\begin{lemma}\label{lem:f}
	For an $n$-symmetric function $f:\R^n \to \R$, $\mathbf F, \mathbf G \in \mathcal M^n$, $\lambda \in[0,1]$ and an $n$-permutation $\pi$, the following hold.
	\begin{enumerate}[(i)]
		\item $\D_f(\mathbf F) = \D_f(\pi(\mathbf F))$.
		\item $\lambda \D_f(\mathbf F) +(1-\lambda) \D_f(\mathbf G) \subset \D_f(\lambda \mathbf F + (1-\lambda)\mathbf G)$. In particular,
		\begin{enumerate}
			\item $\lambda \D_f(\mathbf F) +(1-\lambda) \D_f(\mathbf F) = \D_f( \mathbf F )$.
			\item  $\D_f(\mathbf F)\cap \D_f(\mathbf G) \subset \D_f(\lambda \mathbf F + (1-\lambda)\mathbf G)$.
		\end{enumerate}
	\end{enumerate}
\end{lemma}
\begin{proof}
	(i) holds because of the symmetry of $f$. To prove (ii), for any $H \in \lambda \D_f(\mathbf F) +(1-\lambda) \D_f(\mathbf G)$, there exist $X_1 \sim F_1, \cdots, X_n \sim F_n$, $Y_1\sim G_1, \cdots, Y_n \sim G_n$ and an event $A \in \mathcal F$ independent of $X_1,\cdots, X_n, Y_1,\cdots, Y_n$ such that $\p(A) = \lambda$ and $\(f(X_1, \cdots, X_n) \id_A + f(Y_1, \cdots, Y_n) \id_{A^c}\) \sim H$.
We notice that  \begin{align*} f(X_1, \cdots, X_n) \id_A + f(Y_1, \cdots, Y_n) \id_{A^c}=f\(X_1 \id_A + Y_1 \id_{A^c}, \cdots, X_n \id_A + Y_n \id_{A^c}\),\end{align*}
and $(X_i \id_A + Y_i \id_{A^c}) \sim \lambda F_i + (1-\lambda) G_i$ for any $i = 1, \cdots, n$. Thus we have
 $H\in \D_f(\lambda \mathbf F + (1-\lambda)\mathbf G)$.  This completes the proof of (ii).
\end{proof}

We briefly fix some notation and convention. Let $\Delta_n$ be the standard simplex given by $\Delta_n=\{(\lambda_1,\dots,\lambda_n)\in [0,1]^n:  \sum_{i=1}^n \lambda_i=1\}$. Recall that a doubly stochastic matrix  is a square matrix of nonnegative real numbers, each of whose rows and columns sums to 1 (i.e.~each row or column is in $\Delta_n$).
Denote by $\mathcal Q_n$ the set of $n\times n$ doubly stochastic matrices.
All vectors should be treated as column vectors.
For   $\boldsymbol \lambda=(\lambda_1,\dots,\lambda_n)\in \Delta_n$ and $\mathbf F =(F_1,\dots,F_n)\in \mathcal M^n$, their dot product is
$\boldsymbol\lambda \cdot \mathbf F =  \sum_{i=1}^n  \lambda_iF_i\in \mathcal M$.
For a matrix $\Lambda=(\boldsymbol \lambda_1,\dots,\boldsymbol \lambda_n)^\top\in\mathcal Q_n $  and $\mathbf F  \in \mathcal M^n$, their product is
$\Lambda \mathbf F =  ( \boldsymbol \lambda_1 \cdot \mathbf F,\dots,\boldsymbol\lambda_n \cdot \mathbf F)\in \mathcal M^n$.

The vector $\Lambda \mathbf F$ is a distribution mixture of $\mathbf F$,  and we will call it the \emph{$\Lambda$-mixture} of $\mathbf F$ to emphasize the reliance on $\Lambda$.
Indeed,
 $\Lambda \mathbf F $ can be seen as a vector of weighted averages of $\mathbf F$.
In particular, by choosing $\Lambda=(\frac{1}{n})_{n\times n}$ (here $ (x )_{n \times n}$ means an $ n \times n$ matrix with identical number $x \in \R$),
 we get the vector $(F,\dots,F)$ where $F$ is the average of components of $\mathbf F$. Note that if $\mathbf F\in \M_1^n$, then the mean of any element of $\D_n(\mathbf F)$
 is the same as that of $\D_n(\Lambda  \mathbf F)$.

The first   result below suggests that the   set of aggregation for a tuple of distributions is smaller than that for the weighted averages.
The proof is elementary, but the result allows us to observe the important phenomenon that \emph{more homogeneous marginals lead to a larger aggregation set}.


\begin{theorem}\label{th:1}
For an $n$-symmetric function $f:\R^n \to \R$, $\mathbf F \in \mathcal M^n$ and  $\Lambda\in\mathcal Q_n$,
$\D_f(\mathbf F) \subset \D_f(\Lambda \mathbf F)$. In particular, $\D_n(\mathbf F)\subset \D_n(\Lambda \mathbf F)$.
\end{theorem}
\begin{proof}
Let $\Pi_1,\dots,\Pi_{n!}$  be all different $n$-permutation matrices, i.e. $\Pi_k \mathbf F$ is a permutation of $\mathbf F$.
By Birkhoff's Theorem (Theorem 2.A.2 of \cite{MOA11}),
the set $\mathcal Q_n$ of
doubly stochastic matrices is the convex hull of permutation matrices, that is,
for any $\Lambda \in \mathcal Q_n$,  there exists  $(\lambda_1,\dots,\lambda_{n!})\in \Delta_{n!}$, such that
$$
\Lambda= \sum_{k=1}^{n!}\lambda_k \Pi_k.
$$
Note that  $\mathcal D_f(\mathbf F)= \mathcal D_f(\Pi_k \mathbf F)$ for $k=1,\dots,n!$ by Lemma \ref{lem:f}(i).
Further, by Lemma \ref{lem:f}(ii-b), we have,
$$ \mathcal D_f(\mathbf F)=\bigcap_{k=1}^{n!} \mathcal D_f(\Pi_k \mathbf F)\subset \mathcal D_f\left(\sum_{k=1}^{n!}\lambda_k\Pi_k(\mathbf F)\right) = \mathcal D_f(\Lambda \mathbf F).$$
This completes the theorem.
\end{proof}

As  the sum aggregation is the most common in financial applications, we  will mainly discuss $\D_n$ instead of $\D_f$ in the following context, while keeping in mind that most results on $\D_n$ can be extended naturally to $\D_f$.

\begin{corollary}\label{coro:0}
For $\mathbf F=(F_1,\dots,F_n) \in \mathcal M^n$ and  $\Lambda\in\mathcal Q_n$,
$\D_n(\Lambda \mathbf F)\subset \D_n(F,\dots,F)$ where  $F=\frac{1}n \sum_{i=1}^n F_i$.
 \end{corollary}
By taking $\Lambda$ as the identity in Corollary \ref{coro:0}, we obtain the set inclusion $\D_n(  \mathbf F)\subset \D_n(F,\dots,F)$, which was given in Theorem 3.5 of \cite{BJW14} to find the bounds on $\VaR$ for heterogeneous marginal distributions.

 The doubly stochastic matrices are closely related to majorization order.
 For $\boldsymbol \lambda,\boldsymbol \gamma \in \R^n$,
we say that $\boldsymbol \lambda$ dominates  $\boldsymbol \gamma$  in  \emph{majorization order}, denoted by $ \boldsymbol  \gamma\prec \boldsymbol \lambda$,
if $ \sum_{i=1}^n \phi(\gamma_i)\le \sum_{i=1}^n \phi(\lambda_{i})  $ for all continuous convex functions $\phi$.
 There are several equivalent conditions for this order; see Section 1.A.3 of \cite{MOA11}.
One equivalent condition that is relevant to Theorem \ref{th:1} is that $ \boldsymbol  \gamma\prec \boldsymbol \lambda$ if and only if there exists $\Lambda\in \mathcal Q_n$ such that
$ \boldsymbol  \gamma =\Lambda\boldsymbol \lambda $. We can similarly define majorization order between $\mathbf F,\mathbf G\in \mathcal M^n$,
denoted by $ \mathbf G\prec \mathbf F$,
if $\mathbf G= \Lambda \mathbf F$ for some $\Lambda \in \mathcal Q_n$.
Then, we have the following corollary.

\begin{corollary}\label{coro:1}
For $\mathbf F,\mathbf G\in \mathcal M^n$, if $ \mathbf G\prec \mathbf F$, then
$\D_n(\mathbf F)\subset \D_n(\mathbf G)$.
 \end{corollary}

  \begin{example}[Bernoulli distributions]\label{ex:trivial1}
We apply Theorem \ref{th:1} to Bernoulli distributions.  Let $B_p$ be a Bernoulli cdf with (mean) parameter  $p\in [0,1]$.
Note that a mixture of Bernoulli distributions is still Bernoulli,
and more precisely, for $\mathbf p=(p_1,\dots,p_n)\in [0,1]^n$ and $\mathbf q=(q_1,\dots,q_n)=\Lambda \mathbf p$, we have $\Lambda(B_{p_1},\dots,B_{p_n}) = (B_{q_1},\dots,B_{q_n})$.
Therefore, by Theorem \ref{th:1},
for any $\mathbf p,\mathbf q\in [0,1]^n$ with  $  \mathbf q \prec   \mathbf p$,  we have
$\D_n(B_{p_1},\dots,B_{p_n})\subset \D_n(B_{q_1},\dots,B_{q_n})$.
This result will be used later to discuss joint mixability (see Section \ref{sec:jm}) of Bernoulli distributions.
For instance, we can set $\mathbf p =(0.2,0.8)$, $$\Lambda=\left(\begin{array}{cc}
\frac{1}{4}&\frac{3}{4}\\
\frac{3}{4}&\frac{1}{4}
\end{array}
\right),~\text{and}~\mathbf q= \left(\begin{array}{cc}
\frac{1}{4}&\frac{3}{4}\\
\frac{3}{4}&\frac{1}{4}
\end{array}
\right)(0.2,0.8)=(0.65,0.35).$$
Note that
$\Lambda(B_{0.2}, B_{0.8})=(B_{0.65},B_{0.35}).$ Hence $\D_n(B_{0.2},B_{0.8})\subset \D_n(B_{0.65},B_{0.35}).$
 \end{example}

Next, we discuss how $\Lambda$-mixtures affect the lower sets with respect to convex order. A distribution $F\in \mathcal M_1$ is called smaller than a distribution $G\in \mathcal M_1$  in \emph{convex order}, denoted by $F\lcx G$, if
\begin{equation}\label{eq:convex_order}
  \int \phi \,\d F \le \int \phi \,\d G ~~{\rm for~all~convex}~~\phi:\R\to\R,
\end{equation}
provided that both integrals exist (finite or infinite); see \cite{MS02} and \cite{SS07} for an overview on convex order and the related notion of second-order stochastic dominance.
For a given distribution $F\in \mathcal M_1$, denote by $\C(F)$ the set of all distributions in $\mathcal M_1$ dominated by $F$ in convex order, that is,
$$
\C(F)=   \{G\in \mathcal M_1:G\lcx F\}.
$$
For any distributions $F$ and $G$, we denote by $F\oplus G$ the distribution with quantile function $F^{-1}+G^{-1}$.\footnote{In other words, $F\oplus G$ is the distribution of the sum of  two comonotonic random variables with respective distributions $F$ and $G$. Two random variables $X$ and $Y$ are said to be \emph{comonotonic}, if there exists a random variable $U$ and two increasing functions $f,g$ such that $X=f(U)$ and $Y=g(U)$ almost surely. Such $U$ can be chosen as a standard uniform random variable ($U \sim \mathrm{U}[0,1]$), and $f$ and $g$ can be chosen as the inverse distribution functions of $X$ and $Y$, respectively.}
Moreover, define
$$\C(F_1,\dots,F_n)= \C(F_1\oplus\dots\oplus F_n).$$
The following lemmas give a simple link between the sets $\mathcal D_n$ and $\mathcal C$; see e.g., Lemma 1 of \cite{MWW19}.
\begin{lemma}\label{lm:lm1}
For $\mathbf F \in \mathcal M_1^n$, $\D_n(\mathbf F)\subset \C(\mathbf F)$.
 \end{lemma}

Similarly to the set $ \D_n(\mathbf F)$ in Theorem \ref{th:1}, $\C(\mathbf F)$ also satisfies an order with respect to $\Lambda$-mixture.
 \begin{theorem}\label{th:1p}
For $\mathbf F \in \mathcal M_1^n$ and  $\Lambda\in\mathcal Q_n$,
we have $ \C(\mathbf F) \subset \C(\Lambda \mathbf F)   $.
 \end{theorem}
 \begin{proof}
Note that
 $F_1\oplus \dots \oplus F_n\in \D_n(\mathbf F)$ since
  $F_1\oplus \dots \oplus F_n$ corresponds to the sum of comonotonic random variables with respective distributions $F_1,\dots,F_n$.
 Using Theorem \ref{th:1} and Lemma \ref{lm:lm1},
 we have $\D_n(\mathbf F) \subset \D_n(\Lambda \mathbf F) \subset \C(\Lambda \mathbf F)$. This implies $F_1\oplus \dots \oplus F_n\in \C(\Lambda \mathbf F)$.
 By definition, $ \C(\mathbf F) \subset \C(\Lambda \mathbf F)$.
  \end{proof}

\section{Quantile mixtures}\label{sec:harmo}

In Section \ref{sec:2}, we have seen  a set inclusion  between $\D_n(\mathbf F)$ and $\D_n(\mathbf G)$
 where $\mathbf G$ is a distribution mixture of $\mathbf F$.
 The general message from Theorem \ref{th:1} is that distribution mixtures   enlarge the aggregation sets.
As distribution mixture corresponds to the arithmetic average of distribution functions,
 it would then be of interest to see whether a ``harmonic average" of $F_1,\dots,F_n$ would give similar properties.
 By saying ``harmonic average" of $F_1,\dots,F_n$, we mean the distribution $F$ with $F^{-1}=\frac{1}n \sum_{i=1}^n F^{-1}_i$, i.e., the average of quantiles.
We shall call this type of average as \emph{quantile mixture}.

 In many statistical applications, marginal distributions of a multi-dimensional object are modelled  in the same location-scale family (such as Gaussian, elliptical, or uniform family). The
quantile mixture of such distributions is still in the same family, whereas the distribution mixture is typically no longer in the family.
Moreover, a quantile mixture also corresponds to the combination of comonotonic random variables (such as combining an asset price with a call option on it), and hence finds its natural position in finance.
As such, it is  rather important and practical to consider    quantile mixtures.
\begin{remark} The two types of mixtures are both basic operations on distributions  and often lead to qualitatively very different mathematical results.
As a famous example in decision theory,
the axiom of linearity on distribution mixtures leads to the classic von Neumann-Morgenstern expected utility theory,
whereas  the axiom of linearity on quantile mixtures  leads to the   dual utility theory of \cite{Y87}.

\end{remark}
 For a matrix $\Lambda $ of non-negative elements (not necessarily in $\mathcal Q_n$) and $\mathbf F \in \mathcal M^n$,
let $\Lambda \otimes \mathbf F$
be a vector of distributions $\mathbf G$ such that componentwise, $\mathbf G^{-1}$ is  equal to $ \Lambda \mathbf F^{-1}$.
If $\Lambda\in \mathcal Q_n$, we call $\mathbf G=\Lambda \otimes \mathbf F$  the \emph{$\Lambda$-quantile mixture} of $\mathbf F$.
If $\mathbf F\in \M_1^n$, then the mean of any element of $\D_n(\mathbf F)$
 is the same as that of $\D_n(\Lambda \otimes \mathbf F)$,
similarly to the case of distribution mixture.
This suggests that one may compare $\D_n(\mathbf F)$
with $\D_n(\Lambda \otimes \mathbf F)$, just like what we did in Section \ref{sec:2} for distribution mixture.

The first natural candidates for us to look at are
 $\D_n(F_1,\dots,F_n)$ and $\D_n(F,\dots,F)$ where $F^{-1}=\frac{1}n \sum_{i=1}^n F^{-1}_i$, thus the quantile version of Corollary \ref{coro:0}.  Unfortunately, the sets $\D_n(F_1,\dots,F_n)$ and $\D_n(F,\dots,F)$ are not necessarily comparable, as seen from the following example.
\begin{example}\label{ex:1}
Take $F_1$ as a binary uniform  distribution (with probability $1/2$ at each point) on $\{0,1\}$ and $F_2$ as a binary uniform distribution on $\{0,3\}$.
Clearly, $F$ is a binary uniform distribution on $\{0,2\}$.
$\D_2(F_1,F_2)$ contains distributions supported on $\{0,1,3,4\}$ and $\D_2(F,F)$ contains distributions supported on $\{0,2,4\}$.
Therefore, these two sets do not have a relation of set inclusion.
\end{example}

On the other hand, as a trivial example, if $F_2,\dots,F_n$ are point masses (without loss of generality, we assume that they are point masses at 0),
then $F$ satisfies $F^{-1}=F_1^{-1}/n$.
In this case,  $\D_n(F_1,\dots,F_n)=\{F_1\}\subset \D_n(F,\dots,F)$ holds trivially.
Therefore, we can expect that the inclusion $\D_n(F_1,\dots,F_n)\subset \D_n(F,\dots,F)$ may hold under some special settings.


Below, we note that
 both $\D_n(\mathbf F)$  and $\D_n(\Lambda \otimes \mathbf  F)$   have the same convex-order maximal element. This is in sharp contrast to the case of mixtures in Theorem \ref{th:1p}.
Proposition \ref{prop:qm} can be   verified directly by definition.
\begin{proposition}\label{prop:qm}
For $\mathbf F \in \mathcal M_1^n$ and  $\Lambda\in\mathcal Q_n$,  we have $\C(\mathbf F)=\C(\Lambda \otimes \mathbf  F)$.
\end{proposition}

As we see from Example \ref{ex:1}, $\D_n(\mathbf F)$ and $\D_n(\Lambda \otimes \mathbf F)$ are not necessarily comparable. In \cite{MWW19}, a non-trivial result is established for the aggregation of standard uniform distributions, which leads to an interesting observation along this direction.

\begin{proposition}\label{prop:uniform}
Suppose that $F_1,\dots,F_n$ are uniform distributions, $n\ge 3$, and $\Lambda=(\frac 1n)_{n\times n}$.
Then
$\D_n(\mathbf F) \subset \D_n(\Lambda \otimes \mathbf F)$.
\end{proposition}
\begin{proof}
Note that the components of $\Lambda \otimes \mathbf F$ are uniform distributions with equal length.
By Theorem 5 of \cite{MWW19}, we have
$\D_n(\Lambda \otimes \mathbf F)=\C_n(\Lambda \otimes \mathbf F).$
Using Proposition \ref{prop:qm}, we have
$\C_n(  \mathbf F)=\C_n(\Lambda \otimes \mathbf F)$.
Lemma \ref{lm:lm1} further yields
$\D_n(  \mathbf F)\subset \C_n(  \mathbf F)$. Putting the above results together, we obtain
$\D_n(  \mathbf F)\subset \D_n(\Lambda \otimes \mathbf F)$.
\end{proof}

It is unclear whether  $\D_n(\mathbf F) \subset \D_n(\Lambda \otimes \mathbf F)$
under some other conditions, similarly to Proposition \ref{prop:uniform}.
 Note that the set inclusion  $\D_n(\mathbf F) \subset \D_n(\Lambda \otimes \mathbf F)$
would help us to obtain semi-explicit formulas for bounds on risk measures (such as VaR), since
by choosing $\Lambda=(\frac 1n)_{n\times n}$,
the marginal distributions of
$\Lambda \otimes \mathbf F$
are the same, and  formulas for VaR bounds in e.g., \cite{WPY13} and \cite{BJW14} are applicable; see Section \ref{sec:4}.

There are several sharp contrasts regarding distribution and quantile mixtures.
In addition to the contrast on order relations that we see from Theorem \ref{th:1} and Example \ref{ex:1}, the two notions also treat location shifts on the marginal distributions very differently.
This point will be explained in Section \ref{app:location}.


\section{Bounds on the worst-case values of risk measures}\label{sec:4}

This section is dedicated to exploring the inequalities between the worst-cases value of risk measures in risk aggregation with different marginal distribution tuples.  Our main results in Sections \ref{sec:2} and \ref{sec:harmo} will help to find the inequalities in Proposition \ref{prop:trivial3}.

 \subsection{Risk measures}
We  pay a particular attention to the popular regulatory risk measure VaR, which is a quantile functional. For  $F \in \mathcal M$, for $p\in (0,1)$, define the risk measure $\VaR_p: \mathcal M \to \R$ as
$$
\VaR_p(F) = F^{-1}(p)=  \inf\{x\in \R: F(x)\ge p\}.
$$
Another popular regulatory risk measure is $\ES_p:\mathcal M_1\to \R$ for $p\in (0,1)$, given by
$$
\ES_p(F)=\frac{1}{1-p} \int_p^1 F^{-1}(u) \d u.
$$
Given marginals $\mathbf F$, the worst-case value of VaR in risk aggregation with unknown dependence structure is then defined as
\begin{align*}\overline{\VaR}_p(\mathbf F)& = \sup \{\VaR_p(G): G\in \mathcal D_n(\mathbf F)\}.
\end{align*}
In other words,  $\overline{\VaR}_p(\mathbf F)$ is the largest   value of $\VaR_p$ of the aggregate risk $X_1+\dots+X_n$ over all  possible dependence structures among $X_i\sim F_i,~i=1,\dots,n$.
Similarly, the worst-case value of ES in risk aggregation is  defined as
 $ \overline{\ES}_p(\mathbf F) = \sup \{\ES_p(G): G\in \mathcal D_n(\mathbf F)\}.
 $

The worst-case value of ES in risk aggregation is easy to calculate since ES is consistent with   convex order. On the other hand, worst-case value of VaR in risk aggregation generally does not admit any analytical formula, which is a challenging problem;   results under some specific cases are given in \cite{WPY13}, \cite{PR13} and \cite{BJW14}. To obtain approximations for $\overline{\VaR}_p(\mathbf F)$, one may use the asymptotic equivalence between VaR and ES in \cite{EWW15} and then directly apply ES bounds, or use a numerical algorithm such as the rearrangement algorithm of \cite{PR12} and \cite{EPR13}.

We will discuss a general relationship on  risk measures for different aggregation sets. A \emph{risk measure} is a functional $\rho:\mathcal M_\rho\to \R $, where $\mathcal M_\rho\subset \mathcal M$ is the set of distributions of some financial losses. For instance, if $\rho$ is the mean, then $\mathcal M_\rho$ is naturally chosen as the set of distributions with finite mean.
We denote by $\overline{\rho}(\mathbf F)$ the worst-case value of $\rho$ in risk aggregation for $\mathbf F\in \mathcal M^n$, that is, assuming $ \mathcal D_n(\mathbf F)\subset \mathcal M_\rho$,
\begin{align*} \overline{\rho}(\mathbf F) = \sup \{\rho(G): G\in \mathcal D_n(\mathbf F)\}.
 \end{align*}

\subsection{Inequalities implied by stochastic dominance}
Quite obviously, one can compare the  worst-case values of some risk measures for two tuples of distributions satisfying some stochastic dominance, which we briefly discuss here.

 A distribution $F\in \mathcal M$ is   smaller than a distribution $G$ in \emph{stochastic order} (also first-order stochastic dominance), denoted by $F\lst G$, if $F\geq G$.  For $\mathbf F,\mathbf G\in\mathcal M^n$, we say that $\mathbf F$ is smaller than $\mathbf G$ in stochastic order, denoted by $\mathbf F\lst\mathbf G$, if $F_i\lst G_i,$ $ i=1,\dots,n$. Analogously, for $\mathbf F,\mathbf G\in\mathcal M_1^n$, we say that $\mathbf F$ is smaller than $\mathbf G$ in convex order, denoted by $\mathbf F\lcx\mathbf G$, if $F_i\lcx G_i,$ $ i=1,\dots,n$.

We define two relevant common properties of risk measures. A risk measure $\rho$ is \emph{monotone} if $\rho(F)\le \rho(G)$ whenever $F\lst G$; it is \emph{consistent with convex order} if $\rho(F)\le \rho(G)$ whenever $F\lcx G$. Almost all risk measures used in practice are monotone; ES is consistent with convex order whereas VaR is not. Monetary risk measures (see \cite{FS16}) that are consistent with convex order are characterized by \cite{MW20} and they admit an ES-based representation. In particular, all lower semi-continuous convex risk measures, including ES and expectiles (e.g., \cite{Z16} and \cite{DBBZ16}), are consistent with convex order; we refer to \cite{FS16} for an overview on risk measures.

Now we state in Proposition \ref{prop:trivial5} that one can compare the  worst-case values of some risk measures for $\mathbf F$ and $\mathbf G$ if $\mathbf F$ is smaller than $\mathbf G$ in stochastic order or convex order.
\begin{proposition}\label{prop:trivial5} Let $\rho$ be a risk measure and $\mathbf F, \mathbf G\in\mathcal M^n$ with $\mathcal D_n(\mathbf F), \mathcal D_n(\mathbf G)\subset \mathcal M_\rho$.
\begin{enumerate}[(i)]
\item If $\rho$ is monotone and $\mathbf F\lst\mathbf G$, then
  $\overline{\rho}(\mathbf F)\leq \overline{\rho}(\mathbf G).$
\item If $\rho$ is consistent with convex order and  $\mathbf F\lcx\mathbf G$ with $\mathbf F, \mathbf G\in \mathcal M_1^n$, then
$\overline{\rho}(\mathbf F)\leq \overline{\rho}(\mathbf G).$
  \end{enumerate}
\end{proposition}
\begin{proof} (i) is straightforward to verify. We next focus on (ii).  Since $F_1\oplus\dots\oplus F_n$ is the largest distribution in  $D_n(\mathbf F)$ with respect to convex order  and $\rho$ is consistent with convex order, we have $\overline{\rho}(\mathbf F)=\rho(F_1\oplus\dots\oplus F_n)$. Similarly, $\overline{\rho}(\mathbf G)=\rho(G_1\oplus\dots\oplus G_n)$. Note that
$\mathbf F\lcx\mathbf G$ means $F_i\lcx G_i,~i=1,\dots,n$.
For all $p\in(0,1)$,  using comonotonic-additivity of $\ES_p$,  we have
$$\ES_p(F_1\oplus\dots\oplus F_n)=\sum_{i=1}^{n} \ES_p(F_i)\leq \sum_{i=1}^{n} \ES_p(G_i)=\ES_p(G_1\oplus\dots\oplus G_n),$$
which gives $F_1\oplus\dots\oplus F_n\lcx G_1\oplus\dots\oplus G_n$ (see e.g., Theorem 3.A.5 of \cite{SS07}).
\end{proof}
In the following result, we  will show that the distribution tuples and their $\Lambda$-mixture or $\Lambda$-quantile mixture typically do not satisfy stochastic order or convex order, unless the mixture operation is essentially identical ($\Lambda \mathbf F=\mathbf F$ or $\Lambda\otimes \mathbf F=\mathbf F$). The proof of Proposition \ref{prop:compare} is  put in  Appendix \ref{app:P4}.
\begin{proposition}\label{prop:compare}
Suppose $\Lambda\in \mathcal Q_n$.
The statements within each of (i)-(iv) are equivalent.
  \begin{enumerate}[(i)]
  \item For $\mathbf F\in\mathcal M^n$,
(a) $\Lambda \mathbf F\lst \mathbf F$; (b) $\mathbf F\lst \Lambda \mathbf F$; (c) $\Lambda \mathbf F=\mathbf F$.

  \item For $\mathbf F\in\mathcal M^n$,
  (a) $\Lambda\otimes \mathbf F\lst \mathbf F$;
  (b) $\mathbf F\lst \Lambda\otimes \mathbf F$;
  (c) $\Lambda\otimes \mathbf F=\mathbf F$.
  \item For $\mathbf F\in\mathcal M_1^n$,
  (a) $\Lambda\otimes \mathbf F\lcx \mathbf F$;
  (b) $\mathbf F\lcx\Lambda\otimes \mathbf F$;
   (c) $\Lambda\otimes \mathbf F=\mathbf F$.
   \item For $\mathbf F\in\mathcal M_1^n$,
  (a) $\Lambda \mathbf F\lcx \mathbf F$; (b) $\Lambda \mathbf F=\mathbf F$.
  \end{enumerate}
\end{proposition}

An implication of Proposition \ref{prop:compare} is that the result on stochastic order in Proposition \ref{prop:trivial5} cannot be applied to compare the worst-case values of risk measures for $\mathbf F$ and $\Lambda \mathbf F$ or $\mathbf F$ and $\Lambda\otimes\mathbf F$. Nevertheless, this comparison can be conducted by applying our findings in Sections \ref{sec:2} and \ref{sec:harmo} and some other techniques. This will be the task in the next subsection.

\subsection{Inequalities generated by distribution/quantile  mixtures }\label{sec:mixture_rm}
In the following, we will obtain inequalities between the worst-case values of risk measures for  $\mathbf F$ and $\Lambda \mathbf F$ or $\mathbf F$ and $\Lambda\otimes\mathbf F$. First, we apply Theorem \ref{th:1} and Proposition \ref{prop:qm} and immediately obtain the following result.
\begin{proposition}\label{prop:trivial3}
Let $\rho$ be a risk measure and  $\Lambda\in\mathcal Q_n$.
\begin{enumerate}[(i)]
\item For $\mathbf F \in \mathcal M^n$ with  $ \mathcal D_n(\mathbf F)\subset \mathcal M_\rho$ and $ \mathcal D_n(\Lambda\mathbf F)\subset \mathcal M_\rho$, we have $\overline{\rho}(\mathbf F) \le \overline{\rho}(\Lambda \mathbf F)$;
\item For $\mathbf F \in \mathcal M_1^n$ with  $ \mathcal D_n(\mathbf F)\subset \mathcal M_\rho$ and $ \mathcal D_n(\Lambda\otimes\mathbf F)\subset \mathcal M_\rho$, if $\rho$ is consistent with convex order, then
$\overline{\rho}(\mathbf F) = \overline{\rho}(\Lambda \otimes \mathbf F) =\rho(F_1\oplus \dots \oplus F_n)$.
\end{enumerate}
\end{proposition}
Note that in Proposition \ref{prop:trivial3}, the inequality for distribution mixture  is valid for all risk measures whereas the equality for quantile mixture is constrained to risk measures consistent with convex order.
As $\ES_p$ is a special case of risk measures consistent with convex order, we immediately get $\overline{\ES}_p (\mathbf F) \le \overline{\ES}_p (\Lambda \mathbf F)$ and $\overline{\ES}_p (\mathbf F) = \overline{\ES}_p (\Lambda \otimes \mathbf F)$. Since VaR  is not consistent with convex order, (ii) of Proposition \ref{prop:trivial3}
cannot be  applied to VaR.  Nevertheless, using a recent result  on $\overline{\VaR}$ in \cite{BLLW20},  we obtain an inequality   between $\overline{\VaR}$ for some special marginals  and $\overline{\VaR}$ of their corresponding quantile mixture. Denote by $\mathcal M_D$ (respectively, $\mathcal M_I$) the set of distributions with decreasing (respectively, increasing) densities on their support. Moreover, let $\M_D^n= (\M_D)^n$ and $\M_I^n=(\M_I)^n$.
\begin{theorem}\label{th:scale}
  For $p \in (0,1)$, $\Lambda \in \mathcal Q_n$, and $\mathbf F \in \M_D^n \cup \M_I^n$, we have
  $$\overline{\VaR}_p (\mathbf F )
  \le
  \overline{\VaR}_p (\Lambda \otimes \mathbf{F}).
  $$
\end{theorem}
\begin{proof}
We start with some preliminaries. Define the upper VaR at level $p$ for a cdf $F$ as
  $$
  \VaR_p^*(F) = \inf\{x\in \R: F(x) > p\}, ~~~ p \in (0,1).
  $$
  The worst-case value of the upper VaR in risk aggregation is
 $ \overline{\VaR}_p^*(\mathbf F)  = \sup \{\VaR_p^*(G): G\in \mathcal D_n(\mathbf F)\}.
 $
  For $\mathbf F\in \M_D^n \cup \M_I^n$ and $p \in (0,1)$, Lemma 4.5 of \cite{BJW14} gives
  $$
  \overline{\VaR}_p^*(\mathbf F) = \overline{\VaR}_p(\mathbf F).
  $$
  Using Lemma \ref{lem:LW20} in Appendix \ref{app:a1} (paraphrased from Theorem 2 of \cite{BLLW20}), we have
  \begin{equation}\label{eq:primal}
  \overline{\VaR}_p (\mathbf F) = \inf_{\boldsymbol{\beta} \in \mathbb{B}_n} \sum_{i=1}^{n} \frac{1}{(1-p)(1-\beta)} \int_{p+(1-p)(\beta -\beta_i)}^{1-(1-p)\beta_i} \VaR_u(F_i) \d u,
  \end{equation}
  where $\boldsymbol{\beta}= (\beta_1, \cdots, \beta_n)$, $\beta = \sum_{i=1}^n \beta_i$ and
  $
  \mathbb{B}_n = \{\boldsymbol{\beta} \in [0,1)^n:\beta < 1  \}.
  $
Note that  $$\Lambda \otimes \mathbf{F}\in \M_D^n \cup \M_I^n~~ \text{if}~~ \mathbf F\in \M_D^n \cup \M_I^n.$$ Consequently, for $p\in (0,1)$,
  $$
  \begin{aligned}
  \overline{\VaR}_p (\Lambda \otimes \mathbf{F} ) =& \inf_{\boldsymbol{\beta} \in \mathbb{B}_n} \sum_{i=1}^{n} \frac{1}{(1-p)(1-\beta)} \int_{p+(1-p)(\beta -\beta_i)}^{1-(1-p)\beta_i} \left(\sum_{j=1}^n \Lambda_{i,j} \VaR_u(F_j)  \right) \d u\\
  =&\inf_{\boldsymbol{\beta} \in \mathbb{B}_n}  \sum_{i=1}^n \sum_{j=1}^n \Lambda_{i,j} M_{i,j}(\boldsymbol \beta),
  \end{aligned}
  $$
  where the function $M: \mathbb{B}_n \rightarrow \R^{n\times n}$, mapping an $n$-dimensional vector to an $n \times n$ matrix, is given by
  $$
  M_{i,j}(\boldsymbol{\beta}) =  \frac 1 {(1-p)(1-\beta)} \int_{p+(1-p)(\beta -\beta_i)}^{1-(1-p)\beta_i} \VaR_u(F_j) \d u,~~~i,j=1,\dots,n.
  $$
  We can rewrite \eqref{eq:primal} as
  $$
  \begin{aligned}
  \overline{\VaR}_p (\mathbf{F} )
  =\inf_{\boldsymbol{\beta} \in \mathbb{B}_n}  \sum_{i=1}^n M_{i,i}(\boldsymbol \beta).
  \end{aligned}
  $$
  Let $\Pi_1,\dots,\Pi_{n!}$  be all different $n$-permutation matrices, i.e., $\Pi_k \boldsymbol{\beta}$ is a permutation of $\boldsymbol{\beta}$ for each $k$. By Birkhoff's Theorem (Theorem 2.A.2 of \cite{MOA11}), for $\Lambda \in \mathcal Q_n$, there exists  $(\lambda_1,\dots,\lambda_{n!})\in \Delta_{n!}$ such that
  $
  \Lambda= \sum_{k=1}^{n!}\lambda_k \Pi_k.
  $
  Hence, by writing $\Pi_k \boldsymbol{\beta} = (\beta^k_1,\cdots,\beta^k_n)$ for each $k$, we have
  $$
  \begin{aligned}
  \sum_{i=1}^n \sum_{j=1}^n \Lambda_{i,j} M_{i,j}(\boldsymbol \beta)
  &= \frac{1}{(1-p)(1-\beta)} \sum_{i=1}^{n} \int_{p+(1-p)(\beta -\beta_i)}^{1-(1-p)\beta_i} \left(\sum_{j=1}^n \Lambda_{i,j} \VaR_u(F_j)  \right) \d u\\
  &= \frac{1}{(1-p)(1-\beta)} \sum_{i=1}^{n} \sum_{k=1}^{n!} \lambda_k \int_{p+(1-p)(\beta -\beta^k_i)}^{1-(1-p)\beta^k_i} \VaR_u(F_i)   \d u\\
  &= \sum_{k=1}^{n!} \lambda_k  \sum_{i=1}^{n} \frac{1}{(1-p)(1-\beta)}\int_{p+(1-p)(\beta -\beta^k_i)}^{1-(1-p)\beta^k_i} \VaR_u(F_i)   \d u\\
  &= \sum_{k=1}^{n!} \lambda_k\sum_{i=1}^n M_{i,i}(\Pi_k \boldsymbol \beta).
  \end{aligned}
  $$
Using the above facts, we finally obtain
  $$
  \begin{aligned}
  \overline{\VaR}_p (\mathbf{F})  =\inf_{\boldsymbol{\beta} \in \mathbb{B}_n}  \sum_{i=1}^n M_{i,i}(\boldsymbol \beta)
     &  = \sum_{k=1}^{n!} \lambda_k \inf_{\boldsymbol{\beta} \in \mathbb{B}_n}  \sum_{i=1}^n M_{i,i}(\Pi_k\boldsymbol \beta)\\
  & \leq  \inf_{\boldsymbol{\beta} \in \mathbb{B}_n} \sum_{k=1}^{n!} \lambda_k\sum_{i=1}^n M_{i,i}(\Pi_k \boldsymbol \beta)
=   \inf_{\boldsymbol{\beta} \in \mathbb{B}_n}\sum_{i=1}^n \sum_{j=1}^n \Lambda_{i,j} M_{i,j}(\boldsymbol \beta)   = \overline{\VaR}_p (\Lambda \otimes \mathbf{F}).
  \end{aligned}
  $$
  This completes the proof of the theorem.
\end{proof}

 The restriction of marginals to distributions with monotone densities in Theorem \ref{th:scale} is because of applying Lemma \ref{lem:LW20}. This assumption is common in the literature of VaR bounds (e.g., \cite{WPY13}).  We may expect Theorem \ref{th:scale} to hold for more general classes of $\mathbf F$; this is supported by the numerical results in Figure \ref{fig:power_het}.  Moreover,    for  $\Lambda \in \mathcal Q_n$ and and $\mathbf F \in \M_D^n \cup \M_I^n$, we may expect    $\overline{\rho} (\mathbf F )
  \le
  \overline{\rho} (\Lambda \otimes \mathbf{F})
  $
  for other risk measures $\rho$ than VaR (Theorem \ref{th:scale}) and those consistent with convex order (Proposition \ref{prop:trivial3}).
 Unfortunately, we are unable to prove the above statements in general. Some related open questions are listed in Section \ref{sec:conc}.

 \begin{remark}\label{rem:r1-rw1}
 
The  condition  $\mathbf F\in \M_D^n \cup \M_I^n$  in Theorem \ref{th:scale} can be relaxed to that the $p$-tail distributions of $F_1,\dots,F_n$ are all in $\M_D$ or all in $\M_I$.\footnote{The $p$-tail distribution of $F$ is the distribution of $F^{-1}(U)$
 where $U$ is uniform on $[p,1]$; see e.g., \cite{RU02}.}
  This should be clear since only the $p$-tail distributions are involved in the proof of Theorem \ref{th:scale}. This condition often holds if $p$ is close to $1$, and it allows for Theorem \ref{th:scale} to be applied to many common distributions in risk management.
 
 \end{remark}

 Next, we study location-scale distribution families.
Let $T_x(F) $ be a shift of $F\in \mathcal M$ by adding a constant $x\in \R$ to its location, that is, $T_x(F)$ is the distribution of $X+x$ for $X\sim F$.
For $\mathbf x= (x_1,\dots,x_n)\in \R^n$ and $\mathbf F=(F_1,\dots,F_n)\in \mathcal M^n$, we use the notation
$\mathbf T_{\mathbf x} (\mathbf F) = (T_{x_1}(F_1),\dots,T_{x_n}(F_n)).$
Moreover, for $\lambda \ge 0$, we denote by $F^{\lambda}$ the distribution of $\lambda X$ for  $X\sim F$ and write $\mathbf F^{\boldsymbol \lambda}=(F^{\lambda_1},\dots, F^{\lambda_n})$.

\begin{corollary}\label{item1:th location}  For $p\in (0,1)$, $F \in \mathcal{M}_D \cup \mathcal{M}_I$ , $\boldsymbol \lambda,\boldsymbol \gamma \in \R_+^n$, and $\mathbf x, \mathbf y\in\mathbb{R}^n$, if $ \boldsymbol  \gamma\prec \boldsymbol \lambda$ and $\sum_{i=1}^{n}x_i\leq \sum_{i=1}^{n}y_i$, then
\begin{equation}\label{maj}
\overline{\VaR}_p (\mathbf T_{\mathbf x}(\mathbf F^{\boldsymbol \lambda}))
\le
\overline{\VaR}_p (\mathbf T_{\mathbf y}(\mathbf F^{\boldsymbol \gamma})).
\end{equation}
\end{corollary}
\begin{proof}
By Section 1.A.3 of \cite{MOA11},
 $ \boldsymbol  \gamma\prec \boldsymbol \lambda$ if and only if there exists $\Lambda\in \mathcal Q_n$ such that
$ \boldsymbol  \gamma =\Lambda\boldsymbol \lambda $. This implies
 $\mathbf F^{\boldsymbol \gamma}=\Lambda\otimes\mathbf F^{\boldsymbol\lambda}.$
 By Theorem \ref{th:scale}, it follows that
 $\overline{\VaR}_p (\mathbf F^{\boldsymbol \lambda})\leq \overline{\VaR}_p (\mathbf F^{\boldsymbol \gamma}).$
Moreover, observe that
$$\overline{\VaR}_p (\mathbf T_{\mathbf x}(\mathbf F^{\boldsymbol \lambda}))=\overline{\VaR}_p (\mathbf F^{\boldsymbol \lambda})+\sum_{i=1}^{n}x_i~~\text{and}~~\overline{\VaR}_p (\mathbf T_{\mathbf y}(\mathbf F^{\boldsymbol \gamma}))=\overline{\VaR}_p (\mathbf F^{\boldsymbol \gamma})+\sum_{i=1}^{n}y_i.$$
 By the fact that $\sum_{i=1}^{n}x_i\leq \sum_{i=1}^{n}y_i$, we prove (\ref{maj}).
\end{proof}

\section{Bounds on risk measures for Pareto risk aggregation}\label{sec:bound}

 In this section we study the worst-case risk measure for a portfolio of Pareto risks, and the risk measure is not necessarily consistent with convex order. Throughout this section, we assume that $\rho$ is a monotone risk measure, such as VaR.

 One particular situation  of interest for risk aggregation with non-convex risk measures is when the risks in the portfolio do not have a finite mean.
Note that for a portfolio without finite mean, any non-constant risk measure that is consistent with convex order (including convex risk measures) will have an infinite value. Therefore, one has to use a non-convex risk measure such as
 VaR to assess risks in this situation.

Arguably, the most important class of heavy-tailed risk distributions  is the class of Pareto distributions due to their regularly varying tails and their prominent appearance in extreme value theory; see e.g., \cite{EKM97}.
 A common parameterization of Pareto distributions is given by, for $\theta,\alpha>0$,
	$$
	P_{\alpha,\theta}(x) = 1 -\left(\frac{\theta}{x}\right)^{\alpha},~~x\ge \theta.
	$$
Note that if $X\sim P_{\alpha,1}$, then
$\theta X\sim P_{\alpha,\theta}$, and thus $\theta$ is a scale parameter.
Moreover, the mean of $P_{\alpha,\theta}$ is infinite if and only if $\alpha\in (0,1]$.
 Limited by the current techniques, we confine ourselves to portfolios of risks with a fixed $\alpha$ and possibly different $\theta$.

 For   $\alpha>0 $ and $\boldsymbol \theta =(\theta_1,\dots,\theta_n)\in (0,\infty)^n$,
 let
 $\mathbf P_{\alpha,\boldsymbol \theta} = (P_{\alpha,\theta_1},\dots,P_{\alpha,\theta_n})$.
We are interested in the worst-case value $\overline{\rho}(\mathbf P_{\alpha,\boldsymbol \theta} )$.
We first note some simple properties of the above quantity, which are straightforward to check (a simple proof is put in Appendix \ref{app:Pa}).
\begin{proposition}\label{prop:trivial4}
Let $\rho$ be a monotone risk measure on $\mathcal M$.
 For   $\alpha>0 $ and $\boldsymbol \theta  \in (0,\infty)^n$,
 \begin{enumerate}[(i)]
 \item $\Lambda \otimes \mathbf P_{\alpha,\boldsymbol \theta} = \mathbf P_{\alpha,\Lambda \boldsymbol \theta} $
 for all $\Lambda\in (0,\infty)^{n\times n}$;
\item  $\overline{\rho}(\mathbf P_{\alpha,\boldsymbol \theta} )$  is decreasing  in $\alpha$;
\item $\overline{\rho}(\mathbf P_{\alpha,\boldsymbol \theta} )$  is increasing  in each component of  $\boldsymbol \theta$.
\end{enumerate}
\end{proposition}

The next result contains an ordering relationship on  the aggregation of Pareto risks. In particular, we show that for $\alpha\in (0,1]$, which means the mean of the distribution is infinite, the quantile mixture leads to an even larger worst-case value of risk aggregation than the distribution mixture (this statement is generally not  true for $\alpha>1$; see the figures in Section \ref{sec:numerical}).
This result is not implied by any comparisons obtained in the previous sections, and it seems to be rather specialized for Pareto distributions, as seen from the proof. It is unclear at the moment whether the result can be generalized to other types of distributions without a finite mean.
\begin{theorem}\label{th:mon}
Let $\rho$ be a monotone risk measure on $\mathcal M$.
  For  $\alpha\in (0,1]$, $\boldsymbol \theta =(\theta_1,\dots,\theta_n)\in (0,\infty)^n$,
  and  $\Lambda \in \mathcal Q_n$, we have
 $\overline{\rho} (\mathbf P_{\alpha,  \boldsymbol \theta} )
     \le  \overline{\rho} (\Lambda \mathbf P_{\alpha,  \boldsymbol \theta} )  \le
     \overline{\rho} (    \mathbf P_{\alpha, \Lambda \boldsymbol \theta} ).$
\end{theorem}

\begin{proof} The first inequality follows directly from  Theorem \ref{th:1}.
     Next we focus on the second inequality.
     Recall that $\Lambda=(\boldsymbol \lambda_1,\dots,\boldsymbol \lambda_n)^\top\in\mathcal Q_n $ and let $\boldsymbol \lambda_j=(\lambda_{j,1},\dots,\lambda_{j,n})$ for $j=1,\dots,n$. For any fixed $j\in \{1,\dots,n\}$, denote the cdf of $(\Lambda \mathbf P_{\alpha,  \boldsymbol \theta})_{j} $ by $F_{j}$, then
     $$F_{j}(x)=\sum_{i=1}^{n}\lambda_{j,i}\left(1 -\left(\frac{\theta_{i}}{x}\right)^{\alpha}\right)_{+},~~x\in\mathbb{R}.$$
     For some fixed $x>0$ and $\alpha\in (0,1]$, define  $g(t):=1 -\left(t/x\right)^{\alpha},~t\geq 0$.
      Note that $g$ is a convex function on $[0,
      \infty)$. Hence
     $$F_{j}(x)\geq \sum_{i=1}^{n}\lambda_{j,i}\left(1 -\left(\frac{\theta_{i}}{x}\right)^{\alpha}\right)\geq 1 -\left(\frac{\sum_{i=1}^{n}\lambda_{j,i}\theta_{i}}{x}\right)^{\alpha}.$$
     This implies  $$F_j(x)\geq G_j(x),~ x\geq 0,$$
      where  $ G_{j}=\left(\mathbf P_{\alpha, \Lambda \boldsymbol \theta}\right)_{j}$.
     As $F_{j}\leq_{\rm st} G_{j}$ for $j=1,\dots,n$ and $\rho$ is monotone, by Proposition \ref{prop:trivial5}(i), we have the second inequality.
\end{proof}

Next, we combine the results of
Theorems \ref{th:scale}-\ref{th:mon} and Propositions  \ref{prop:trivial3}-\ref{prop:trivial4}  with a special focus on
 $\VaR_p$, $p\in (0,1)$.  The proof is straightforward and omitted.

 \begin{proposition}\label{cor:VaR_pareto}
For $p\in (0,1)$, $\boldsymbol \theta =(\theta_1,\dots,\theta_n)\in (0,\infty)^n$,
and  $\Lambda \in \mathcal Q_n$,
\begin{enumerate}[(i)]
\item If $\alpha\in(0,\infty)$,       $\overline{\VaR}_p (\mathbf P_{\alpha,  \boldsymbol \theta} )
     \le  \overline{\VaR}_p (\Lambda \mathbf P_{\alpha,  \boldsymbol \theta} ) ;$

\item If $\alpha\in(0,\infty)$, $\overline{\VaR}_p (\mathbf P_{\alpha,  \boldsymbol \theta} )
\le
\overline{\VaR}_p (\mathbf P_{\alpha, \Lambda \boldsymbol \theta} );$

\item If $\alpha\in(0,1]$,       $\overline{\VaR}_p (\mathbf P_{\alpha,  \boldsymbol \theta} )
     \le  \overline{\VaR}_p (\Lambda \mathbf P_{\alpha,  \boldsymbol \theta} )  \le
     \overline{\VaR}_p (    \mathbf P_{\alpha, \Lambda \boldsymbol \theta} ).$
\end{enumerate}
 \end{proposition}
 Proposition \ref{cor:VaR_pareto} is useful  for the application in Section \ref{sec:7} on multiple hypothesis testing, where $P^r$ follows a Pareto distribution for a p-value $P$ and $r<0$.
 Some further properties of $\overline{\VaR}_p (\mathbf P_{\alpha,  \boldsymbol \theta} )
$ are put in Appendix \ref{app:Pa2}.

\section{Numerical illustration}\label{sec:numerical}

Define a $3\times 3$ doubly stochastic matrix by
\begin{equation}\label{eq:matrix}
\Lambda=0.8 \times I_3 + 0.2 \times \left(\frac{1}{3}\right)_{3 \times 3},
\end{equation}
where $I_3$ is the $3\times 3$ identity matrix.
In this section, we consider a sequence of doubly stochastic matrices $\{\Lambda^k\}_{k\in\mathbb{N}}$
to numerically illustrate the ordering relationships and inequalities obtained throughout the paper. Note that $\Lambda^k$ is more ``homogeneous" as $k$ grows larger, and $\Lambda^k\to (\frac 13)_{3\times3}$ as $k\to\infty$.
The general messages obtained from the numerical examples are listed as follows.

\begin{enumerate}
\item For general marginals,  the value of $\overline{\VaR}$ becomes larger after making a distribution mixture (Proposition \ref{prop:trivial3}(i)); this is shown in all figures.

\item For marginals with monotone densities, with a quantile mixture, the value of $\overline{\VaR}$ becomes larger (Theorem \ref{th:scale}); see Figures \ref{fig:power_pareto}-\ref{fig:power_het_family}. Numerical examples in Figure \ref{fig:power_het} indicate that  Theorem \ref{th:scale} may also  hold for marginals  with non-monotone densities. Nevertheless, the order does not hold for arbitrary marginals. A counterexample, involving discrete marginals, is provided in Figure \ref{fig:power_het_bino10}.

\item For Pareto distributions with infinite mean, the value of $\overline{\VaR}$ of the quantile mixture is larger than that of the distribution mixture (Proposition \ref{cor:VaR_pareto}(iii)); see Figure \ref{fig:power_pareto}(b). This relationship does not hold for Pareto distributions with finite mean;  see Figure \ref{fig:power_pareto}(a).
\end{enumerate}

\subsection{Illustration of theoretical results}

\begin{figure}[htbp]
	\centering
	\subfigure[Pareto distribution with finite mean ($\alpha = 3$)]{
		\label{fig:power_homo}
		\includegraphics[width=0.4\textwidth]{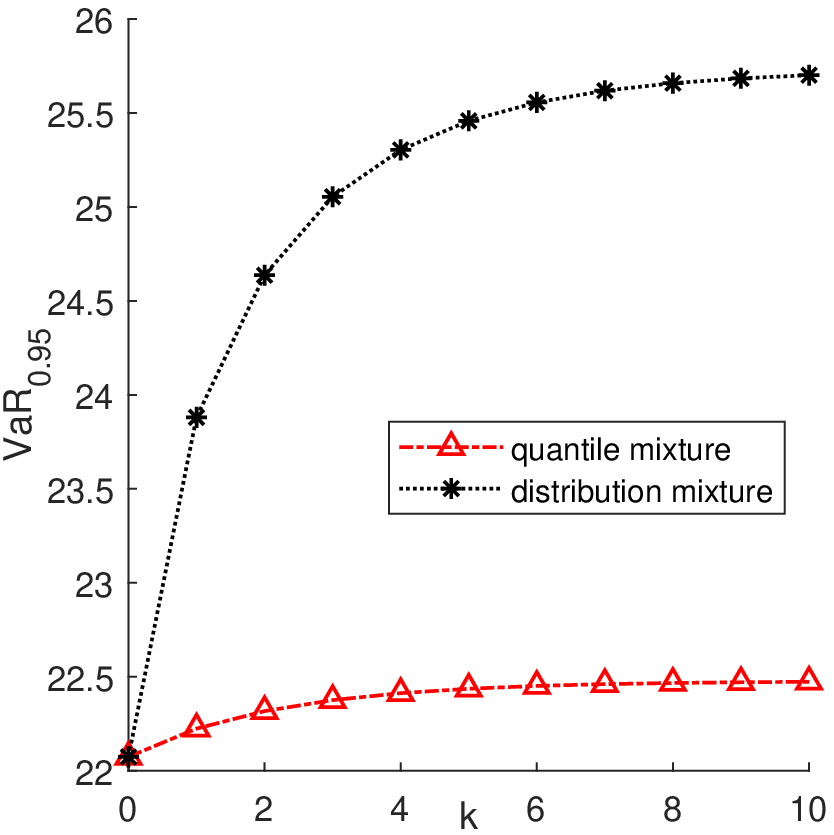}}
	\subfigure[Pareto distribution with infinite mean ($\alpha = 1/3$)]{
		\label{fig:power_homo_infi}
		\includegraphics[width=0.4\textwidth]{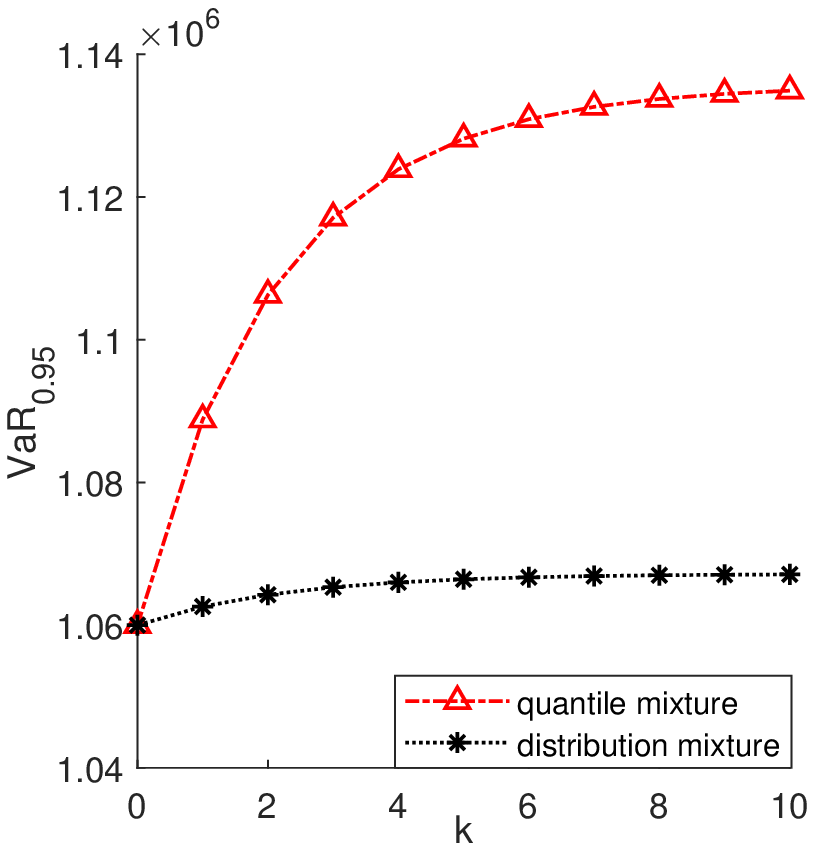}}
	\caption{Quantile mixture: $\overline{\VaR}_{p} (\Lambda^k \otimes \mathbf{P}_{\alpha, \boldsymbol{\theta}}) = \overline{\VaR}_{p}(\mathbf P_{\alpha, \Lambda^k \boldsymbol \theta})$; Distribution mixture: $\overline{\VaR}_{p} (\Lambda^k \mathbf{P}_{\alpha, \boldsymbol{\theta}})$. Setting: $p=0.95$; $\boldsymbol{\theta} = (1, 2, 3)$, $X_i \sim \text{Pareto}(\alpha, \theta_i)$, $i=1,2,3$; $\Lambda$ is defined by \eqref{eq:matrix}; $k = 0, 1, \cdots, 10$. }
	\label{fig:power_pareto}
\end{figure}

In this subsection, we discuss marginals with monotone densities ($\mathbf{F} \in \M_D^n \cup \M_I^n$). We have $\Lambda^k \otimes \mathbf{F}, \Lambda^k  \mathbf{F} \in \M_D^n \cup \M_I^n$. According to Lemma \ref{lem:LW20} in Appendix \ref{app:a1}, we obtain a formula (Equation \eqref{eq:primal}) for $\mathbf{F}$, $\Lambda^k\otimes \mathbf{F}$ and $\Lambda^k \mathbf{F}$ and numerically compute the exact values for $\overline{\VaR}$.

In Figure \ref{fig:power_pareto}, we  consider Pareto distributions with finite mean ($\alpha = 3$) and infinite mean ($\alpha = 1/3$), respectively. The ordering relationships in Proposition \ref{cor:VaR_pareto}(i)-(ii) for Pareto distributions with the same $\alpha$ are visualized  as the curves in Figure \ref{fig:power_pareto} are all increasing in $k$. In Figure \ref{fig:power_pareto}(b), it turns out that for the case with infinite mean the quantile mixture gives larger value of $\overline{\VaR}$ than that given by the distribution mixture. This coincides with the conclusion in Proposition \ref{cor:VaR_pareto}(iii). Interestingly, we observe from Figure \ref{fig:power_pareto}(a) that the value of $\overline{\VaR}$ given by distribution mixture is larger than the one with quantile mixture, which
is contrary to  the case with infinite mean (Figure \ref{fig:power_pareto}(b)). It is an open question whether this conclusion is true for general doubly stochastic matrices $\Lambda$ and all $\alpha>1$.

We next focus on Pareto distributions with different $\alpha$ in Figure \ref{fig:power_hetero}. First observe that the curves of quantile mixture and distribution mixture in Figure \ref{fig:power_hetero} are both increasing in $k$, which is consistent with Theorem \ref{th:scale} and Proposition \ref{prop:trivial3}(i). 
Comparing the two curves, it is shown that value for the distribution mixture in this case is smaller than the one for quantile mixture.
\begin{figure}[htbp]
  \centering
  \includegraphics[width=0.7\textwidth]{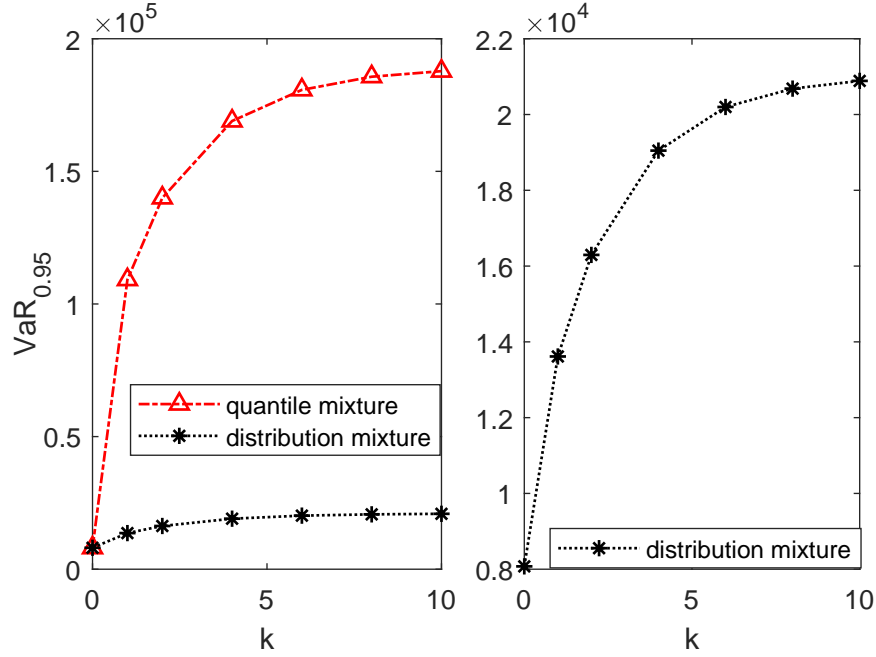}
  \caption{
    Quantile mixture: $\overline{\VaR}_{p} (\Lambda^k \otimes \mathbf{F})$; Distribution mixture: $\overline{\VaR}_{p} (\Lambda^k \mathbf{F})$. Setting: $p=0.95$; $\boldsymbol{\alpha} = (1/3, 4, 5)$, $\boldsymbol{\theta} = (1, 2, 3)$, $X_i \sim \text{Pareto}(\alpha_i, \theta_i)$, $i=1,2,3$; $\Lambda$ is defined by \eqref{eq:matrix}; $k = 0, 1, 2, 4, 6, 8, 10$. The right panel zooms in on the range of the distribution mixture.}
  \label{fig:power_hetero}
\end{figure}

Heterogeneous distribution families with decreasing densities are considered in Figure \ref{fig:power_het_family}.  As we can see, the curves  are both increasing in Figure \ref{fig:power_het_family}, which coincides with the statements in Theorem \ref{th:scale} and Proposition \ref{prop:trivial3}(i). We can also observe that the value for distribution mixture is smaller than the corresponding one for quantile mixture in Figure \ref{fig:power_het_family}, which is the same as it has been shown in Figure \ref{fig:power_hetero}.
\begin{figure}[htbp]
  \centering
  \includegraphics[width=0.6\textwidth]{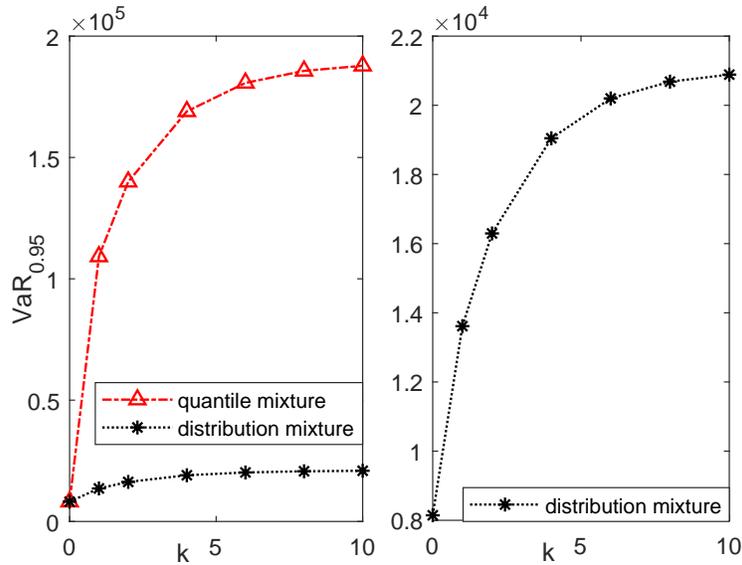}
  \caption{Quantile mixture: $\overline{\VaR}_{p} (\Lambda^k \otimes \mathbf{F})$; Distribution mixture: $\overline{\VaR}_{p} (\Lambda^k \mathbf{F})$.  Setting: $p=0.95$; $X_1 \sim \text{Pareto}(1/3, 1)$, $X_2 \sim \Gamma(1,2)$, $X_3 \sim \text{Weibull}(1,1/2)$; $\Lambda$ is defined by \eqref{eq:matrix}; $k = 0, 1, 2, 4, 6, 8, 10$. The right panel zooms in on the range of the distribution mixture.}
  \label{fig:power_het_family}
\end{figure}

\subsection{Conjectures for general distributions}
Explicit expressions for $\overline{\VaR}_p(\mathbf F)$ are unavailable for general marginal distributions. Fortunately, we can approximate the value of $\overline{\VaR}_p(\mathbf F)$  using the rearrangement algorithm (RA) of \cite{EPR13} and get an upper bound on $\overline{\VaR}_p(\mathbf F)$ using \eqref{eq:primal-app} in Lemma \ref{lem:LW20}. 

For distributions with non-monotone densities including Gamma and Weibull, the  curves of both distribution and quantile mixtures in Figure \ref{fig:power_het} are increasing in $k$. The result  on distribution mixture is consistent with  Proposition \ref{prop:trivial3}(i), and the result  on quantile mixture seems  to suggest that the conclusion in Theorem \ref{th:scale} may be valid for more general distributions with non-monotone densities.
This conjectured extension of Theorem \ref{th:scale} would hold if \eqref{eq:primal} holds for more general distributions, which is a difficult question.


\begin{figure}[htbp]
	\centering
	\subfigure[$X_3 \sim \text{Pareto}(3, 1)$]{
		\label{fig:power_het_notdec}
		\includegraphics[width=0.4\textwidth]{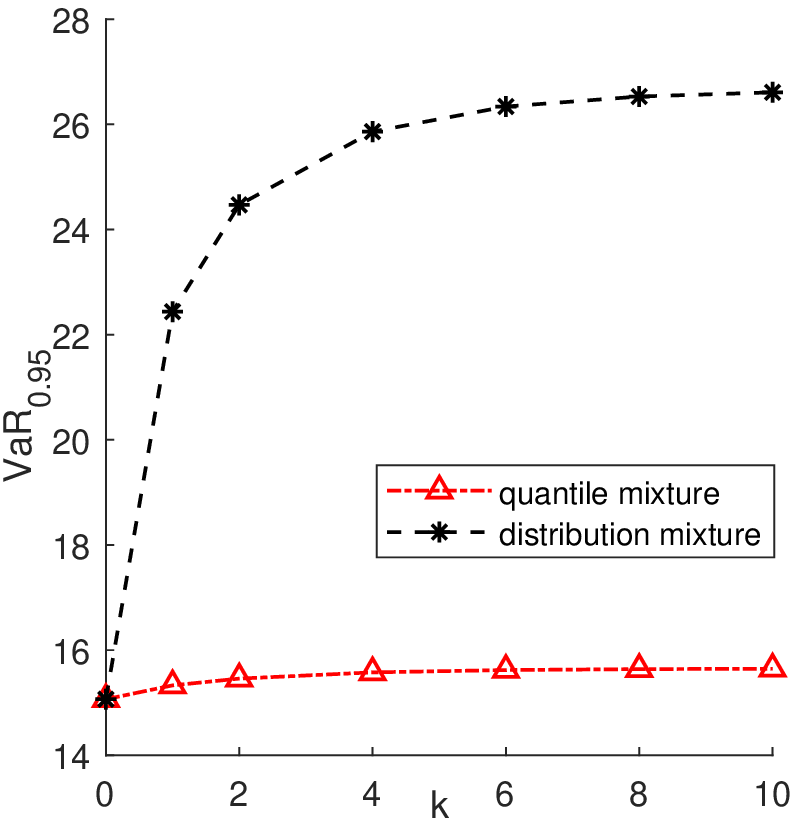}}
	\subfigure[$X_3 \sim \text{LogNormal}(0, 1)$]{
		\label{fig:power_het_logn}
		\includegraphics[width=0.4\textwidth]{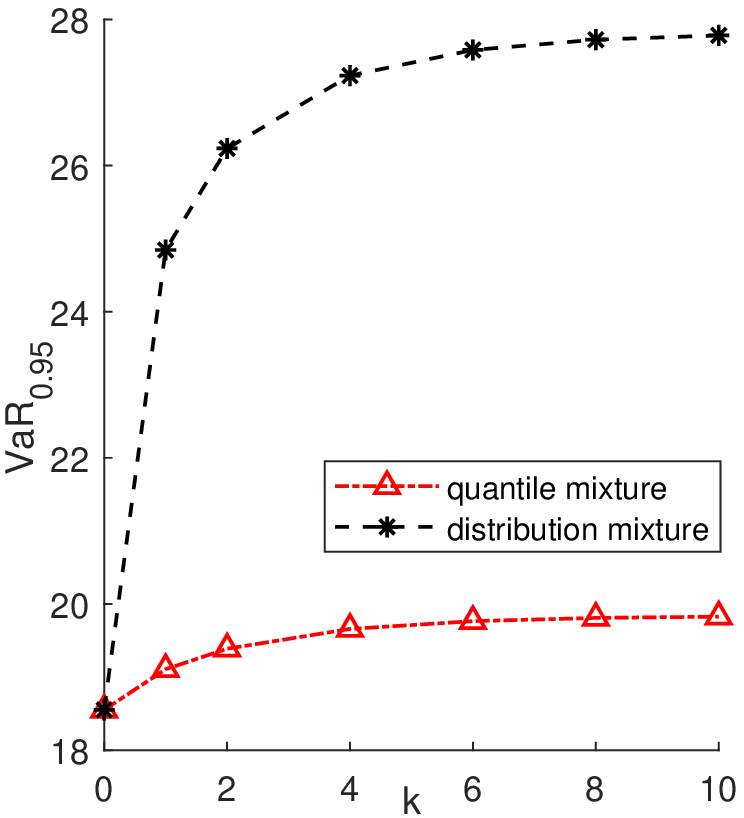}}
	\caption{Quantile mixture: $\overline{\VaR}_{p} (\Lambda^k \otimes \mathbf{F})$; Distribution mixture: $\overline{\VaR}_{p} (\Lambda^k \mathbf{F})$.  Setting: $p=0.95$; $X_1 \sim \Gamma(5,1)$, $X_2 \sim \text{Weibull}(1,5)$, left panel: $X_3 \sim \text{Pareto}(3, 1)$, right panel: $X_3 \sim \text{LogNormal}(0, 1)$; $\Lambda$ is defined by \eqref{eq:matrix}; $k = 0, 1, 2, 4, 6, 8, 10$. }
	\label{fig:power_het}
\end{figure}

 The above observation is no longer true for discrete distributions. We observe in Figure \ref{fig:power_het_bino10} that the curve of the quantile mixture is not increasing at some points (in this example, we have chosen a small $p=0.01$ for illustration). This shows that the claim in Theorem \ref{th:scale} cannot be extended to arbitrary, in particular discrete, distributions.

\begin{figure}[htbp]
	\centering
	\includegraphics[width=0.42\textwidth]{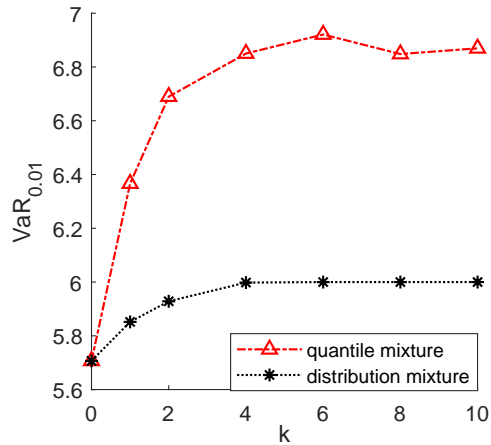}
	\caption{Quantile mixture: $\overline{\VaR}_{p} (\Lambda^k \otimes \mathbf{F})$;  Distribution mixture: $\overline{\VaR}_{p} (\Lambda^k \mathbf{F})$.  Setting: $p=0.01$; $X_1 \sim \text{Binomial}(10, 0.1)$, $X_2 \sim \Gamma(5,1)$, $X_3 \sim \text{Weibull}(1,5)$; $\Lambda$ is defined by \eqref{eq:matrix}; $k = 0, 1, 2, 4, 6, 8, 10$.}
	\label{fig:power_het_bino10}
\end{figure}

\section{Applications}\label{sec:application}
\subsection{ Portfolio diversification with dependence uncertainty}\label{sec:portfolio}

We discuss applications of our results to portfolio diversification in the presence of dependence uncertainty.
In this section, we treat risk measures as  functionals on the space of random variables,
that is, for a random variable $X$ and a risk measure $\rho$, we write $\rho(X)=\rho(F)$ if $X\sim F$.

For tractability, we consider a simple setting where the vector of losses $(X_1,\dots,X_n)$ has identical marginal distributions $F$.
A classic portfolio selection problem is to choose a portfolio position   $\boldsymbol \lambda = (\lambda_1,\dots,\lambda_n)\in \Delta_n$ to minimize
\begin{equation}
\label{eq:r1-rw1}
R_\rho(\boldsymbol \lambda) := \rho \left(\sum_{i=1}^n \lambda_i  X_i\right).
\end{equation}
Alternatively, one may consider an objective which involves both risk and return, such as maximizing
the quantity $
\E [- \sum_{i=1}^n \lambda_i  X_i ]-
\alpha \rho  (\sum_{i=1}^n \lambda_i  X_i )
$
for some $\alpha >0$ (e.g., $\alpha$ may arise as a Lagrangian multiplier);  in our setting, this problem is equivalent to \eqref{eq:r1-rw1} since   $\E [ \sum_{i=1}^n \lambda_i  X_i ]$ is  constant over  $\boldsymbol \lambda \in \Delta_n$.
Intuitively, for two portfolio positions $\boldsymbol \lambda$ and  $\boldsymbol \gamma$, we can say that  $\boldsymbol \gamma $  is more diversified than  $\boldsymbol \lambda$   if $ \boldsymbol  \gamma  \prec \boldsymbol  \lambda $, since in this case $\boldsymbol \gamma$ can be obtained from averaging components of $ \boldsymbol  \lambda$, i.e., $\boldsymbol \gamma = \Lambda \boldsymbol \lambda$ for some  $\Lambda\in \mathcal Q_n$ (see Section \ref{sec:2}).
Due to diversification effect,  one may expect, under the assumption that the marginal distributions of $(X_1,\dots,X_n)$ are identical, 
\begin{equation}
R_\rho(\boldsymbol \gamma)\le R_\rho(\boldsymbol \lambda) \mbox{~~~if $\boldsymbol \gamma $  is more diversified than  $\boldsymbol \lambda$.}\label{eq:r1-rw3}
\end{equation}
Note that $(\frac 1n,\dots,\frac 1n) \prec  \boldsymbol  \lambda \prec (1,0,\dots,0)$
for any portfolio position $\boldsymbol \lambda$, meaning that the most diversified portfolio is the equally weighted one, and the least diversified portfolio is concentrating on a single source of risk.

To compute the value of $R_\rho (\boldsymbol \lambda)$ in \eqref{eq:r1-rw1} requires a full specification of the joint distribution of $(X_1,\dots,X_n)$. In the presence of dependence uncertainty, we may take a worst-case approach by minimizing
\begin{equation}
\label{eq:r1-rw2}
\overline{R}_\rho(\boldsymbol \lambda) :=\sup\left\{ \rho \left(\sum_{i=1}^n \lambda_i  Y_i\right):Y_1,\dots,Y_n\sim F\right\}.
\end{equation}
Under the setting of optimizing  \eqref{eq:r1-rw2}, our intuition is that diversification should not yield any benefit, since the portfolio may not have any diversification effect due to unknown dependence; see \cite{WZ20} for discussions on the absence of diversification effect within the Fundamental Review of the Trading Book from the Basel Committee on Banking Supervision (\cite{BASEL19}). Hence, one may expect, as the marginal distributions are identical, that
\begin{equation}
\overline{R}_\rho(\boldsymbol \gamma)= \overline{R}_\rho(\boldsymbol \lambda)  \mbox{~~~even if $\boldsymbol \gamma $  is more diversified than  $\boldsymbol \lambda$ (in fact, for all $\boldsymbol \gamma $ and $\boldsymbol \lambda $).}\label{eq:r1-rw4}
\end{equation}
A similar observation is made in Proposition 1 of \cite{PP18}, which says that for a subadditive, comonotonic-additive and positively homogeneous risk measure, diversification under dependence uncertainty does not decrease the aggregate risk. These assumptions  on the risk measure are not necessary for our result below.

The next proposition, based on Theorem \ref{th:scale} and Proposition \ref{prop:trivial3}, shows that, under some extra conditions, the two intuitive equations \eqref{eq:r1-rw3} and \eqref{eq:r1-rw4}  hold for   risk measures  consistent with convex order. For VaR, one arrives at  a statement in the reverse direction: the more diversified  portfolio has a larger risk under dependence uncertainty.
\begin{proposition}\label{prop:r1-rw1}
Suppose that $ \boldsymbol  \gamma  \prec \boldsymbol  \lambda $, $(X_1,\dots,X_n)$ has identical marginal distributions $F$ with finite mean, and $\rho$ is a risk measure.
\begin{enumerate}[(i)]
\item If $\rho$ is consistent with convex order and $(X_1,\dots,X_n)$ is exchangeable,\footnote{A random vector $\mathbf X$ is exchangeable if $\mathbf X$ is identically distributed as $\pi(\mathbf X)$ for any permutation $\pi$.} then  $R_\rho(\boldsymbol \gamma)\le R_\rho(\boldsymbol \lambda) $.
\item If $\rho$  is consistent with convex order, then $\overline{R}_\rho(\boldsymbol \gamma)= \overline{R}_\rho(\boldsymbol \lambda)$.
\item If $\rho=\VaR_p$ for some $p\in (0,1)$ and $F\in \M_D\cup\M_I$, then   $\overline{R}_\rho(\boldsymbol \gamma)\ge  \overline{R}_\rho(\boldsymbol \lambda)$.
\end{enumerate}
Moreover, in (i) and (iii), the inequalities are generally not equalities.
\end{proposition}
\begin{proof}
Write $\boldsymbol \gamma =(\gamma_1,\dots,\gamma_n)$
and  $\boldsymbol \lambda =(\lambda_1,\dots,\lambda_n)$.
Take $X\sim F$, and let
 $\mathbf F$ and $\mathbf G$ be the tuples of marginal distributions of $(\gamma_1X,\dots,\gamma_nX)$ and $(\lambda_1X,\dots,\lambda_nX)$, respectively.
Using $ \boldsymbol  \gamma  \prec \boldsymbol  \lambda $,
 there exists $\Lambda\in \mathcal Q_n$ such that
$\mathbf F = \Lambda \otimes \mathbf G$.
Since $X_1,\dots,X_n\sim F$,  $\mathbf F$  and $\mathbf G$ are also tuples of marginal distributions of $(\gamma_1X_1,\dots,\gamma_nX_n)$ and $(\lambda_1X_1,\dots,\lambda_nX_n)$, respectively. Hence, we have
\begin{align}\label{eq:r1-rw5} \overline R_\rho(\boldsymbol \gamma)= \overline{\rho}(\mathbf F)= \overline{\rho}(\Lambda \otimes \mathbf G) \mbox{~~~and~~~}\overline R_\rho(\boldsymbol \lambda)= \overline{\rho}(\mathbf G).
\end{align}
\begin{enumerate}[(i)]
\item As $ \boldsymbol  \gamma  \prec \boldsymbol  \lambda $ and $(X_1,\dots,X_n)$ is exchangeable, by Theorem 3.A.35 of \cite{SS07}, we have $\sum_{i=1}^n \gamma_i X_i \prec_{\rm cx} \sum_{i=1}^n \lambda_i X_i$.
Hence, $R_\rho(\boldsymbol \gamma)\le R_\rho(\boldsymbol \lambda) $.
The inequality is strict when, for instance, $\rho=\ES_{0.5}$, $X_1,\dots,X_n$ are iid normal, $\boldsymbol \gamma =(\frac 1 n,\dots,\frac 1n)$, and $\boldsymbol \lambda =(1,0,\dots,0)$.
\item This follows directly from Proposition \ref{prop:trivial3}(ii) and \eqref{eq:r1-rw5}.
\item  The inequality $\overline{R}_\rho(\boldsymbol \gamma)\ge  \overline{R}_\rho(\boldsymbol \lambda)$ follows directly from  Theorem \ref{th:scale} and \eqref{eq:r1-rw5}.
The inequality is strict in, for instance, the situation of Figure \ref{fig:power_pareto}(a), where $F$ is a Pareto distribution with $\alpha=3$. \qedhere
\end{enumerate}
\end{proof}
We make a few observations from Proposition \ref{prop:r1-rw1}. For  identical marginal distributions in $\M_D$ or $\M_I$, under dependence uncertainty,
 VaR  yields a bigger risk if the portfolio is more diversified. This may be seen as another disadvantage of VaR, which is well known to be problematic regarding diversification.
In contrast, any risk measure  consistent with convex order, such as   ES,
would simply ignore diversification effect in this setting (where diversification benefit is unjustifiable).
Moreover, without dependence uncertainty, for an exchangeable vector of losses, a risk measure  consistent with convex order   rewards diversification,
and there is no such general relationship for VaR.
For the inequality in Proposition \ref{prop:r1-rw1} (iii), it suffices to require the $p$-tail distribution of $F$ to be in  $ \M_D\cup\M_I$; see Remark \ref{rem:r1-rw1}.

\subsection{Merging p-values in hypothesis testing}
 \label{sec:7}

In this subsection, we apply our results to p-merging methods following the setup of \cite{VW19}. A random variable $P$ is a \emph{p-variable} if $\mathbb{P}(P\leq\epsilon)\leq\epsilon$ for all $\epsilon\in(0,1)$, and its realization is called a p-value. In multiple hypothesis testing, one natural problem is to merge individual p-values into one p-value.  More specifically,  with $n$ p-variables $P_{1},\dots,P_{n}$, one needs to choose an increasing Borel function $F:[0,1]^{n}\rightarrow [0,\infty)$ as a \emph{merging function} such that  $F(P_{1},\dots,P_{n})$ is a p-variable.
$F$ is a \emph{precise merging function} if
for each $\epsilon\in(0,1)$,
$\mathbb{P}(F(P_1,\dots,P_n)\leq\epsilon)=\epsilon$ for some p-variables $P_1,\dots,P_n$.

As explained in \cite{VW19}, an advantage of using averaging methods to combine p-values, compared to classic methods on order statistics,
is that we can introduce weights to p-values in an intuitive way.
 Without imposing any dependence assumption on the individual p-variables, an averaging method  uses, for $r\in[-\infty,\infty]$ ($r\in\{-\infty,0,\infty\}$ are interpreted as limits),
 $$F:[0,1]^n\to [0,\infty), ~(p_{1},\dots,p_{n})\mapsto a_{r,\mathbf w}(w_{1}p_{1}^{r}+\dots+w_{n}p_{n}^{r})^{\frac{1}{r}},$$
 as the merging function, where $a_{r,\mathbf w}$ is a constant multiplier  and $\mathbf{w}=(w_{1},\dots,w_{n})\in\Delta_n$. The constant $a_{r,\mathbf w}$ is chosen so that  $F$ is a precise merging function, thus the most powerful choice of the constant multiplier. Let $\mathcal{U}$ be the set of uniform random variables distributed on [0,1]. Lemma 1 in \cite{VW19}  gives
 \begin{equation*}
a_{r,\mathbf{w}} =
\begin{cases}
-\sup\left\{q_{0}(-\sum_{i=1}^{n}w_{i}P_{i}^{r})\mid P_{1},\dots,P_{n}\in\mathcal{U}\right\}^{-1/r}, & r>0;\\
\exp\left(\sup\{q_{0}(\sum_{i=1}^{n}w_{i}\log(1/P_{i}))\mid P_{1},\dots,P_{n}\in\mathcal{U}\}\right),& r = 0;\\
\sup\left\{q_{0}\left(\sum_{i=1}^{n}w_{i}P_{i}^{r}\right)\mid P_{1},\dots,P_{n}\in\mathcal{U}\right\}^{-1/r},&  r<0,
\end{cases}
\end{equation*}
where $q_0: X\mapsto \inf\{x\in \R: \p(X\le x)>0\}$ is the essential infimum.
Clearly,  $a_{r,\mathbf{w}}$ involves calculating $\overline{\VaR}_p(\mathbf F)$ for Pareto, exponential or Beta distributions, and letting $p\downarrow 0$.

 Denote $a_{r,\mathbf w}$ by $a_{r,n}$ where $\mathbf w=(1/n,\dots,1/n)$.
Analytical results for $a_{r,n}$ has been well studied in \cite{VW19} whereas results for  $a_{r,\mathbf w}$ are limited since there are no analytical formulas of $\overline{\VaR}_p(\mathbf F)$ in general for heterogeneous marginal distributions. Although the rearrangement algorithm of \cite{PR12} and \cite{EPR13} can be used to calculate $a_{r,\mathbf w}$ numerically, the calculation burden becomes quite heavy in high-dimensional situation, which is unfortunately very common in multiple hypothesis testing. It turns out that our Theorem \ref{th:scale} is helpful to provide a convenient upper bound on $a_{r,\mathbf w}$.

\begin{proposition}\label{prop:p_value}
For $r\in \R$, we have $a_{r,\mathbf{w}}\leq a_{r,n}.$
\end{proposition}
\begin{proof}
Note that for $r<0$, $P_{i}^{r}$, $i=1,\dots,n$, has a decreasing density,
and  $ (1/n,\dots, 1/n)\prec (w_{1},\dots,w_{n})$ in majorization order.
By letting $p\downarrow 0$ in Proposition \ref{cor:VaR_pareto}, we have
$$\sup\left\{q_{0}\left(\sum_{i=1}^{n}w_{i}P_{i}^{r}\right)\mid P_{1},\dots,P_{n}\in\mathcal{U}\right\}
\le
\sup\left\{q_{0}\left(\sum_{i=1}^{n}\frac 1nP_{i}^{r}\right)\mid P_{1},\dots,P_{n}\in\mathcal{U}\right\}.$$
Therefore $a_{r,\mathbf{w}}\leq a_{r,n}$ for $r<0$. If $ r\ge 0$, the argument can be proved similarly using Corollary \ref{item1:th location}.
\end{proof}
The interpretation of Proposition \ref{prop:p_value} is that, when using a weighted p-merging method, one can safely rely on the same coefficient obtained from a symmetric p-merging method. This is particularly convenient when validity of the test is more important than  the quality of an approximation; see \cite{VW19} for more discussions on such applications.


\section{Some further technical discussions} \label{sec:jm}

\subsection{Location shifts for distribution and quantile mixtures}
\label{app:location}

In this section we discuss the difference between distribution and quantile mixtures when location shifts are applied.
Let $V_x=\{ (x_1,\dots,x_n)\in \R^n: x_1+\dots+x_n=x\}$ for $x\in \R$.
For $\mathbf F\in \mathcal M^n$ and $\mathbf x\in V_x$,     we have the invariance relation
\begin{equation}\mathcal D_n (\mathbf T_{\mathbf x} (\mathbf F))=   T_x (\mathcal D_n (\mathbf F)).\label{eq:cond1}\end{equation}
 The aggregation set of quantile mixture is invariant under location shifts of the marginal distributions,  in sharp contrast to the case of distribution mixture.
For $\mathbf F\in \mathcal M^n$ and $\mathbf x\in V_x$, it holds  that  for $\Lambda \in \mathcal Q_n$,
$$\mathcal D_n ( \Lambda \otimes \mathbf T_{\mathbf x} (\mathbf F)) =T_{x} \left(\mathcal D_n (\Lambda \otimes \mathbf F)\right).$$
That means, $\mathcal D_n ( \Lambda \otimes \mathbf T_{\mathbf x} (\mathbf F))$ is the same for all $\mathbf x\in V_x$.
However, this does not hold for the distribution mixture, that is, generally, $\mathcal D_n ( \Lambda   \mathbf T_{\mathbf x} (\mathbf F))$
is not the same for $\mathbf x\in V_x$, and
$$\mathcal D_n ( \Lambda   \mathbf T_{\mathbf x} (\mathbf F)) \ne T_{x} \left(\mathcal D_n (\Lambda   \mathbf F)\right).$$
In particular, for $x\ne 0$ and $F_1\ne F_2$,
$$\mathcal D_2 \left( \frac 12  (T_x(F_1)+F_2),\frac 12  (T_x(F_1)+F_2)\right) \ne \mathcal D_2 \left( \frac 12  (F_1+T_x(F_2)), \frac 12  (F_1+T_x(F_2)))\right) .$$
The above example shows that distribution mixture and quantile mixtures treat location shifts differently.

Inspired by the above observation, we slightly generalize Theorem \ref{th:1} by including location shifts.
For $\mathbf F\in \mathcal M^n$, we define the set $\mathcal A_n(\mathbf F)$ of averaging and location shifts of $\mathbf F$ as
$$
\mathcal A_n(\mathbf F)=\{\Lambda \mathbf T_{\mathbf x} (\mathbf F): \Lambda \in \mathcal Q_n, ~\mathbf x\in \R^n,~x_1+\dots+x_n=0\},
$$
and denote by $\overline{\mathcal A_n(\mathbf F)}$ the closure of the convex hull of $\mathcal A_n(\mathbf F)$ with respect to weak convergence. It is straightforward to check
$$\overline{\mathcal A_n(\mathbf T_{\mathbf y}(\mathbf F))}=\mathbf T_{\mathbf y}\left(\overline{\mathcal A_n(\mathbf F)}\right),~~~\mathbf y=(y,\dots,y)\in \mathbb{R}^n.$$

\begin{proposition}\label{prop:closure}
For  $\mathbf F\in \mathcal M^n$ and  $\mathbf G\in\overline{\mathcal A_n(\mathbf F)}$, we have
$
\mathcal D_n(\mathbf F)\subset \mathcal D_n (\mathbf G).
$
\end{proposition}
\begin{proof}
First, by Theorem \ref{th:1} and \eqref{eq:cond1},
$\mathcal D_n(\mathbf F)\subset \mathcal D_n (\mathbf G)$ for each $\mathbf G \in \mathcal A_n(\mathbf F)$.
Denote by $\mathrm{cx}(\mathcal A_n(\mathbf F))$ the convex hull of $\mathcal A_n(\mathbf F)$.
By Lemma \ref{lem:f}(ii-b),
for each $\mathbf G\in \mathrm{cx}(\mathcal A_n(\mathbf F))$,  we have $\mathcal D_n(\mathbf F)\subset \mathcal D_n (\mathbf G)$.
Take $\mathbf G\in
\overline{\mathcal A_n(\mathbf F)}$, and write
it as the limit of $\{\mathbf G_k\}_{k=1}^\infty\subset \mathrm{cx}(\mathcal A_n(\mathbf F))$.
It follows that for any $F\in \mathcal D_n(\mathbf F)$,
$F$ is also in $\mathcal D_n( \mathbf G_k)$.
This implies $F$ is also in $\mathcal D_n( \mathbf G)$ by the compactness property in Theorem 2.1(vii-b) of \cite{BJW14}.
\end{proof}

\subsection{Connection to joint mixability}

Joint mixability  (\cite{WPY13} and \cite{WW16}) is a central concept in the study of risk aggregation with dependence uncertainty, and analytical results are quite limited. In this section, we study the implication of our results on conditions for joint mixability. We denote by  $\delta_x$ the point mass at $x\in \R$.

\begin{definition}[Joint mixability]\label{def1}
  An $n$-tuple of distributions $\mathbf F\in \mathcal M^n$ is \emph{jointly mixable} (JM) if $\mathcal D_n(\mathbf F)$ contains a point mass distribution $\delta_x$, where $x\in \R$ is called a \emph{center} of $\mathbf F$.
\end{definition}

Example \ref{ex:trivial1} implies a conclusion on the joint mixability  of Bernoulli distributions.
\begin{proposition}\label{prop:bernoulli}
  For $p_1,\dots,p_n\in [0,1]$,
  $(B_{p_1},\dots,B_{p_n})$ is jointly mixable if and only if $\sum_{i=1}^n p_i$ is an integer.
\end{proposition}
\begin{proof}
  The ``only-if" part is trivial since the sum of Bernoulli random variables takes value in integers.
  To show the ``if" part, let  $k=\sum_{i=1}^n p_i$ and  $\mathbf 1_k\in \{0,1\}^n$ be a vector whose first $k$ entries are 1 and the remaining entries are 0.
  It is clear that $\mathbf p\prec \mathbf 1_k$ (see Section 1.A.3 of \cite{MOA11}).
  Hence, from Example \ref{ex:trivial1}, $$\{\delta_{k}\}= \mathcal D_n( \underbrace{B_{1},\dots,B_1}_{k},\underbrace{B_0,\dots,B_0}_{n-k})\subset \mathcal D_n(B_{p_1},\dots,B_{p_n}).$$
  Therefore  $(B_{p_1},\dots,B_{p_n})$ is jointly mixable.
\end{proof}

The set $ \overline{\mathcal A_n(\mathbf F)}$ 
can also be used to obtain    joint mixability of some tuples of distributions. In particular, we shall see in the following proposition that
$ \overline{\mathcal A_n(\delta_0,\dots,\delta_0)}$ is the set of all jointly mixable tuples with center $0$.
\begin{proposition}\label{th:2}
  For
  $\mathbf G\in \mathcal M^n$, the following statements are equivalent.
  \begin{enumerate}[(i)]
    \item $\mathbf G$ is jointly mixable.
    \item $\mathbf G\in \overline{\mathcal A_n(\delta_c,\dots,\delta_c)}$ for some $c\in \R$.
    \item $\mathbf G\in \overline{\mathcal A_n(\mathbf F)}$ for some $\mathbf F\in \mathcal M^n$ which is jointly mixable.
  \end{enumerate}
\end{proposition}

\begin{proof}
  (ii)$\Rightarrow$(iii) is trivial. (iii)$\Rightarrow$(i):
  Suppose that $\mathbf G\in \overline{\mathcal A_n(\mathbf F)}$ and   $\mathbf F$ is jointly mixable with center $x\in \R$.
  By Proposition \ref{prop:closure}, we have
  $\{\delta_{x}\}\subset \mathcal D_n(\mathbf F) \subset \mathcal D_n(\mathbf G).$
  This shows $\mathbf G$ is jointly mixable.
  Next, we show (i)$\Rightarrow$(ii).
  Suppose that $\mathbf G$ is jointly mixable, and without loss of generality we can assume it has center $0$. By definition, there exists a random vector $\mathbf X=(X_1,\dots,X_n)$ such that $X_i\sim G_i$ and $X_1+\dots+X_n=0$.
  Denote by $H$ the distribution measure of $\mathbf X$.
  For $A \in \mathcal B(\R)$ and $i=1,2,\dots, n$,
  $$
  G_i(A)=\p(X_i\in A)=\int_{\R^n} \p(X_i\in A|\mathbf X=\mathbf y) H(\d  \mathbf y)=\int_{\R^n} \delta_{y_i}(A)   H(\d \mathbf y),
  $$
  and as a consequence,
  $$\mathbf G(A)=(G_1(A),\dots,G_n(A))= \int_{\R^n} (\delta_{y_1}(A),\dots,\delta_{y_n}(A))  H( \d \mathbf y).$$
  Noting that $H$ is supported in $V_0=\{(y_1,\dots,y_n)\in \R^n:y_1+\dots+y_n=0\}$, we have
  $$
  \mathbf G(A)= \int_{V_0} (\delta_{y_1}(A),\dots,\delta_{y_n}(A))  H( \d \mathbf y)  =  \int_{V_0}\mathbf T_{\mathbf y} (\delta_0(A),\dots,\delta_0(A)) H(\d \mathbf y).
  $$
  Hence, we conclude that $ \mathbf G \in  \overline{\mathcal A_n(\delta_0,\dots,\delta_0)}. $
\end{proof}

The set $\overline{\mathcal A_n(\delta_c,\dots,\delta_c)}$ is quite rich and cannot be analytically characterized.
The simple example of uniform distributions might be helpful to understand Proposition \ref{th:2}.
Suppose that $F_i=\mathrm{U}[0,a_i]$, $a_i>0$, $i=1,\dots,n$, and $\sum_{i=1}^n a_i \ge 2 \bigvee_{i=1}^n a_i$.
By Theorem 3.1 of \cite{WW16}, we know that $\mathbf F$ is jointly mixable.
Then, Proposition \ref{th:2} implies that every   tuple in the set
$\overline{\mathcal A_n(\mathbf F)}$ is jointly mixable.

It remains an open question whether it is possible to characterize the set  $\overline{\mathcal A_n(\mathbf F)}$
for uniform random variables. This would lead to many classes of jointly mixable distributions including those with monotone densities and symmetric densities; see \cite{WW16}.

\section{Concluding remarks and open questions}\label{sec:conc}
This paper studies the ordering relationship for aggregation sets where the marginal distributions for different sets are connected by either a distribution mixture  or a  quantile mixture. For general marginal distributions, the aggregation set becomes larger after making a distribution mixture  on the marginal risks, whereas the aggregation sets are not necessarily comparable in general by a quantile mixture on the marginal risks. Nevertheless, we obtain several useful results especially on the comparison of VaR aggregation, which has applications in and outside financial risk management.

Although   the marginal distributions are assumed known in our main setting,   this assumption is not essential for the interpretation of our results in practical situations.
In case  both marginal uncertainty and dependence uncertainty are present, our results can be directly applied to obtain  ordering relationships, as we explain below. Suppose that $\Lambda\in \mathcal Q_n$ and  $\mathcal F\subset \mathcal M^n$ is a set of possible marginal  models, representing   uncertainty on the marginal distributions.
In this case, the set of all possible distributions of aggregate risk is $\bigcup_{\mathbf F\in \mathcal F}  \D_n(\mathbf F)$,
and the worst-case value of a risk measure $\rho$ is
$
  \sup\{\rho(G): G\in \D_n(\mathbf F), ~\mathbf F\in \mathcal F\} =\sup_{\mathbf F\in \mathcal F} \overline{\rho} (\mathbf F ) .
$
Using Theorem \ref{th:1}, Proposition \ref{prop:trivial3} and Theorem \ref{th:scale}, we have
$$
\bigcup_{\mathbf F\in \mathcal F}  \D_n(\mathbf F)\subset \bigcup_{\mathbf F\in \mathcal F}   \D_n(\Lambda   \mathbf F),
\mbox{~~~~}
\sup_{\mathbf F\in \mathcal F} \overline{\rho} (\mathbf F )   \le \sup_{\mathbf F\in \mathcal F}  \overline{\rho} (\Lambda  \mathbf{F}),
$$
and, if $ \mathcal F \subset \M_D^n \cup \M_I^n$,
  $$\sup_{\mathbf F\in \mathcal F} \overline{\VaR}_p (\mathbf F )
  \le
 \sup_{\mathbf F\in \mathcal F}  \overline{\VaR}_p (\Lambda \otimes \mathbf{F}).
  $$
  Thus, our results on  set inclusion and   risk measure inequalities   remain valid in the presence of marginal uncertainty.

Many questions on quantile mixtures are still open, and we conclude the paper with four of them.
The first question concerns whether $\D_n(\mathbf F)\subset \D_n(\Lambda \otimes \mathbf F)$ holds for cases other than the uniform distributions in Proposition \ref{prop:uniform}.
As we have seen from Example \ref{ex:1}, for $\mathbf F\in \mathcal M^n$ and $\Lambda \in \mathcal Q_n$,
$\D_n(\mathbf F)$ and $\D_n(\Lambda \otimes \mathbf F)$
are generally not comparable.
It remains open whether  $\D_n(\mathbf F) \subset \D_n(\Lambda \otimes \mathbf F)$
under some conditions.
For instance,  Proposition \ref{prop:uniform} requires $n\ge 3$ and $\Lambda$ being a constant times the identity, to use the characterization of $\D_n(\mathbf F)$ from \cite{MWW19}.
It remains unclear whether the same conclusion holds for $n=2$ or other choices of $\Lambda$.

The second   question concerns decreasing densities (or increasing densities).
A concrete conjecture is presented below, which is inspired by  Theorem \ref{th:scale}.
It is unclear how to formulate natural classes of distributions other  than $\mathcal M_D$ (or $\mathcal M_I$)  such that  similar statements can be expected.


%
%
%
%
%
%

\begin{conjecture}\label{conj4}
  For   $\Lambda \in \mathcal Q_n$ and $\mathbf F \in \mathcal M_D^n$, we have
  $   \D_n(\mathbf F) \subset \D_n( \Lambda \otimes \mathbf F)$.
  Weaker versions of this conjecture are:
  \begin{enumerate}[(i)]
    \item For  $F \in \mathcal M_D$, and $\boldsymbol \lambda,\boldsymbol \gamma \in \R_+^n$, if
    $ \boldsymbol  \gamma\prec \boldsymbol \lambda$, then   $   \D_n(F^{\lambda_1},\dots,F^{\lambda_n}) \subset \D_n(F^{\gamma_1},\dots,F^{\gamma_n})$.
    \item For $F_1,\dots,F_n \in \mathcal M_D$, $\D_n(F_1,\dots,F_n)\subset \D_n(F,\dots,F)$ where $F^{-1}=\frac{1}n \sum_{i=1}^n F^{-1}_i$.
    \item  For $F  \in \mathcal M_D$ and $(\lambda_1,\dots,\lambda_n)\in \Delta_n$, $\D_n(F^{n\lambda_1},\dots,F^{n\lambda_n})\subset \D_n(F ,\dots,F )$.
  \end{enumerate}
\end{conjecture}
It is obvious that
the main statement in
Conjecture \ref{conj4}
implies (i)  by noting that one can choose $\Lambda$ such that
$\boldsymbol \gamma =\Lambda \boldsymbol \lambda$ and it implies (ii)  by choosing $\Lambda=(\frac 1n)_{n\times n}$.
Both (i) and (ii) imply (iii).  An example is provided below to illustrate the connection of Conjecture \ref{conj4} to joint mixability.
\begin{example}
We make a connection of Conjecture \ref{conj4} to Theorem 3.2 of \cite{WW16}, which says that
for  $F_i \in \mathcal M_D$ with essential support $[0,b_i]$, $i=1,\dots,n$,
$\D_n(F_1,\dots,F_n)$ contains a point mass if and only if the \emph{mean-length condition} holds, that is,
$$\sum_{i=1}^n \mu_i \ge \max_{i=1,\dots,n} b_i $$
where $\mu_i$ is the mean of $F_i$, $i=1,\dots,n$.
For   $\Lambda \in \mathcal Q_n$ and $\mathbf F \in \mathcal M_D^n$, let $(\hat \mu_1,\dots,\hat \mu_n)$ be the mean vector of
$\Lambda \otimes \mathbf F$. Note that
 $$\sum_{i=1}^n \hat \mu_i = \mathbf 1_n^\top \Lambda \boldsymbol \mu = \mathbf 1_n^\top \boldsymbol \mu =   \sum_{i=1}^n \mu_i,$$
 where $\mathbf 1_n=(1,\dots,1)\in\mathbb{R}^n$.
 On the other hand, each component of $\Lambda \otimes \mathbf F$ has a shorter  or equal length of  support than the maximum length of $\mathbf F$.
As a consequence, if the mean-length condition holds for $ \mathbf F$, then it also holds for $\Lambda \otimes \mathbf F$.
Therefore, if
 $\D_n(\mathbf F) $ contains a point mass, then so does  $\D_n(\Lambda \otimes \mathbf F) $;
 on the contrary, if $\D_n(\Lambda \otimes \mathbf F) $ contains a point mass,  $\D_n(\mathbf F) $ does not necessarily contains a point mass, since it may have a longer length of the maximum support.
 This, at least intuitively, suggests that
  $\D_n(\mathbf F)  \subsetneq \D_n(\Lambda \otimes \mathbf F) $ may hold, as in Conjecture \ref{conj4}.
  \end{example}

The third question is about the order of VaR for quantile mixture.
Our numerical results in Figure \ref{fig:power_het}
suggest that the VaR relation
  $$\overline{\VaR}_p (\mathbf F )
  \le
  \overline{\VaR}_p (\Lambda \otimes \mathbf{F})
  $$
    holds for more general choices of $\mathbf F$ than the ones in Theorem \ref{th:scale}.
 We are not sure what general conditions on $\mathbf F$ will guarantee this relation to hold.

The last question concerns a cross comparison of distribution and quantile mixtures. As we see from Proposition \ref{cor:VaR_pareto},
 $$  \overline{\VaR}_p (\Lambda\mathbf F )  \le
     \overline{\VaR}_p (\Lambda  \otimes  \mathbf F  ) $$ holds for $\mathbf F$ being a vector of Pareto distributions with the same shape parameter and infinite mean.
     We wonder whether the same relationship holds for other distributions without a finite mean. Note that for the case of finite mean, the relationship may be reversed, as illustrated in Figure \ref{fig:power_pareto}; however we do not have a proof for the reverse inequality (assuming finite mean) either.
     Generally, it is unclear to us whether and in which situation $\D_n (\Lambda\mathbf F )$ and $\D_n (\Lambda\otimes \mathbf F )$  are comparable.

\subsection*{Acknowledgements}
We thank an Associate Editor and two anonymous referees for helpful comments on an earlier version of the paper.
Y.~Liu is financially supported by the China Scholarship Council. R.~Wang
acknowledges financial support from the Natural Sciences and Engineering Research Council of Canada
(NSERC, RGPIN-2018-03823, RGPAS-2018-522590) and from the Center of Actuarial Excellence Research
Grant from the Society of Actuaries.

\subsection*{AUTHOR CONTRIBUTION}
All authors contributed equally to the paper.

\subsection*{DATA AVAILABILITY STATEMENT}
Data sharing not applicable - no new data generated.
\appendix
\normalsize

\section{Some proofs and further technical results}

\subsection{A lemma used in the proof of Theorem \ref{th:scale}}
\label{app:a1}
The following lemma is rephrased from Theorem 2 of \cite{BLLW20}.
\begin{lemma}\label{lem:LW20}
 For $p \in (0,1)$  and  any $\mathbf F=(F_1, \cdots, F_n) \in \M^n$,
  \begin{equation}\label{eq:primal-app}
  \overline{\VaR}_p^* (\mathbf F) \le \inf_{\boldsymbol{\beta} \in \mathbb{B}_n} \sum_{i=1}^{n} \frac{1}{(1-p)(1-\beta)} \int_{p+(1-p)(\beta -\beta_i)}^{1-(1-p)\beta_i} \VaR_u(F_i) \d u,
  \end{equation}
  where $\boldsymbol{\beta}= (\beta_1, \cdots, \beta_n)$, $\beta = \sum_{i=1}^n \beta_i$ and
  $
  \mathbb{B}_n = \{\boldsymbol{\beta} \in [0,1)^n:\beta < 1  \},
  $
  and the above inequality is an equality if $\mathbf F\in \M_D^n \cup \M_I^n$.
\end{lemma}

\subsection{Proof of Proposition \ref{prop:compare}}
\label{app:P4}
\begin{proof} We first focus on (i). We will show (a) $\Leftrightarrow$ (c). (c) $\Rightarrow$ (a) is trivial by the definition of stochastic order. For (a) $\Rightarrow$ (c), note that $\Lambda \mathbf F\lst \mathbf F$ with $\Lambda=(\Lambda_{ij})$ implies
\begin{align}\label{inequalites}
\sum_{j=1}^n\Lambda_{ij}F_j(x)\geq F_i(x),~x\in\mathbb{R},i=1,\dots,n.
\end{align}
Adding all the inequalities in (\ref{inequalites}) yields
\begin{align*}
 \sum_{i=1}^{n}\sum_{j=1}^n\Lambda_{ij}F_j(x)\geq \sum_{i=1}^nF_i(x),~x\in\mathbb{R}.
\end{align*}
Due to the fact that $\Lambda$ is a doubly stochastic matrix, we have
\begin{align*}
 \sum_{i=1}^{n}\sum_{j=1}^n\Lambda_{ij}F_j(x)=\sum_{i=1}^nF_i(x),~x\in\mathbb{R}.
\end{align*}
Hence all the inequalities in (\ref{inequalites}) are essentially equalities.
This  proves (c). We can analogously show that (b) $\Leftrightarrow$ (c). This establishes the claims in (i).
We will omit the proof of (ii) since it is  similar to the proof of (i).

\noindent
We next focus on (iii). Trivially, (c) $\Rightarrow$ (a) and (c) $\Rightarrow$ (b). Next, we will only show (a) $\Rightarrow$ (c) since (b) $\Rightarrow$ (c) is similar. Denote by $\mathbf G=(G_1,\dots,G_n)=\Lambda\otimes \mathbf F$. Hence
\begin{align*}
  G_i^{-1}=\sum_{j=1}^{n}\Lambda_{ij}F_j^{-1}.
\end{align*}
By definition, $\Lambda\otimes \mathbf F\lcx \mathbf F$ implies
$G_i\lcx F_i, i=1,\dots, n$. It is well known
(see e.g., Theorem 3.A.5 of \cite{SS07}) that for any two distributions  $F$ and $G$ in $\mathcal M_1$,
\begin{align} \label{eq:equivalence}
 F\lcx G  ~~\Leftrightarrow ~~
\ES_p(F)\le \ES_p(G)~ \mbox{for all} ~p\in (0,1).
\end{align}
Moreover, by the comonotonic-additivity of $\ES_p$, we have
\begin{align*}
  \ES_p(G_i)=\sum_{j=1}^{n}\Lambda_{ij}\ES_p(F_j),~i=1,\dots,n.
\end{align*}
Consequently,
\begin{align}\label{inequalites1}
  \ES_p(G_i)=\sum_{j=1}^{n}\Lambda_{ij}\ES_p(F_j)\leq \ES_p(F_i),~p\in (0,1), i=1,\dots, n.
\end{align}
Noting that $\Lambda$ is a doubly stochastic matrix, similarly as in the proof of (i), adding  all the inequalities in (\ref{inequalites1}) leads to
\begin{align*}
  \sum_{i=1}^{n}\ES_p(G_i)=
  \sum_{i=1}^{n}\sum_{j=1}^{n}\Lambda_{ij}\ES_p(F_j)=\sum_{i=1}^{n} \ES_p(F_i),~p\in (0,1).
\end{align*}
This implies that the inequalities in (\ref{inequalites1}) are equalities, which means that $\Lambda\otimes \mathbf F=\mathbf F$ by (\ref{eq:equivalence}). We complete the proof of (iii).

\noindent
Finally, we consider (iv). (b) $\Rightarrow$ (a) is trivial. We will show (a) $\Rightarrow$ (b). By (\ref{eq:equivalence}), $\Lambda \mathbf F\lcx \mathbf F$ is equivalent to
\begin{align}\label{inequalities2}
\ES_p(F_i)\geq \ES_p\left(\sum_{j=1}^{n}\Lambda_{ij}F_j\right),~~~i=1,\dots,n.
\end{align}
Moreover, by the concavity of  $\ES_p$ on mixtures (e.g., Theorem 3 of \cite{WWW20}), we have
$$\ES_p\left(\sum_{j=1}^{n}\Lambda_{ij}F_j\right)\ge  \sum_{j=1}^n \Lambda_{ij} \ES_p(F_j).$$
Therefore, we have
\begin{align}\label{inequalities3}
\ES_p(F_i)\geq \ES_p\left(\sum_{j=1}^{n}\Lambda_{ij}F_j\right)\geq \sum_{j=1}^n \Lambda_{ij} \ES_p(F_j),~~~i=1,\dots,n.
\end{align}
Adding the inequalities in (\ref{inequalities3}) with noting that $\Lambda$ is a doubly stochastic matrix  yields
\begin{align*}
\sum_{i=1}^{n}\ES_p(F_i)\geq \sum_{i=1}^{n}\ES_p\left(\sum_{j=1}^{n}\Lambda_{ij}F_j\right)\geq \sum_{i=1}^n\sum_{j=1}^n \Lambda_{ij} \ES_p(F_j)=\sum_{i=1}^{n}\ES_p(F_i).
\end{align*}
Hence
\begin{align*}
\sum_{i=1}^{n}\ES_p(F_i)=\sum_{i=1}^{n}\ES_p\left(\sum_{j=1}^{n}\Lambda_{ij}F_j\right),
\end{align*}
which implies that inequalities in (\ref{inequalities2}) are all equalities.  We establish the claim by (\ref{eq:equivalence}). \end{proof}

\subsection{Proof of Proposition \ref{prop:trivial4}}
\label{app:Pa}

\begin{proof}
\begin{enumerate}[(i)]
\item Note that $(P_{\alpha,\theta})^{-1}=\theta (P_{\alpha,1})^{-1}$ for $\theta, \alpha>0$. Hence we prove (i) by showing that $$(\Lambda \otimes \mathbf P_{\alpha,\boldsymbol \theta})^{-1}=\Lambda (\mathbf P_{\alpha,\boldsymbol \theta})^{-1}=\mathbf P_{\alpha,\Lambda\boldsymbol\theta}^{-1}.$$
\item Let $\mathcal{U}$ be the set of uniform random variables on $[0,1]$. By monotonicity of $\rho$, we have, for $0<\alpha_1<\alpha_2$,
     \begin{align*}
     \overline{\rho} (\mathbf P_{\alpha_1,  \boldsymbol \theta} )=&\sup\left\{\rho\left(\theta_{1}U_{1}^{-1/\alpha_1}+\dots+\theta_{n}U_{n}^{-1/\alpha_1}\right)\mid U_{1},\dots,U_{n}\in\mathcal{U}\right\}\\
      \geq&\sup\left\{\rho\left(\theta_{1}U_{1}^{-1/\alpha_2}+\dots+\theta_{n}U_{n}^{-1/\alpha_2}\right)\mid U_{1},\dots,U_{n}\in\mathcal{U}\right\}= \overline{\rho} (  \mathbf P_{\alpha_2,  \boldsymbol \theta} ).
     \end{align*}
     This implies that $\overline{\rho} (\mathbf P_{\alpha,  \boldsymbol \theta} )$ is decreasing in $\alpha$.
\item By monotonicity of $\rho$, we can establish the claim of (iii) similarly as the proof of (ii). \qedhere
\end{enumerate}
\end{proof}

\subsection{Some further properties of $\overline{\VaR}_p(\mathbf P_{\alpha,\boldsymbol \theta})$}
\label{app:Pa2}
Properties of $\overline{\rho}(P_{\alpha,\boldsymbol \theta})$  in Proposition \ref{prop:trivial4} can be strengthened for $\rho=\VaR_p$.
\begin{proposition}\label{prop:trivial2}

 For  $p\in (0,1)$, $\alpha>0 $ and $\boldsymbol \theta  \in (0,\infty)^n$,
 \begin{enumerate}[(i)]
 \item $\overline{\VaR}_p(\mathbf P_{\alpha,\boldsymbol \theta} )$ is increasing  and continuous in $p$;
\item  $\overline{\VaR}_p(\mathbf P_{\alpha,\boldsymbol \theta} )$  is decreasing and continuous in $\alpha$;
\item $\overline{\VaR}_p(\mathbf P_{\alpha,\boldsymbol \theta} )$  is increasing and continuous in each component of  $\boldsymbol \theta$;
\item  $\overline{\VaR}_p(\mathbf P_{\alpha,\boldsymbol \theta} )$  is homogeneous in $\boldsymbol \theta$, that is, for $\lambda > 0$,
$$
\overline{\VaR}_p (\mathbf P_{\alpha,\lambda \boldsymbol \theta} ) = \lambda \overline{\VaR}_p (\mathbf P_{\alpha,\boldsymbol \theta} );
$$
\item If $\alpha>1$, then
\begin{equation}\label{eq:supVaR_Pareto}
\frac{\mathbf 1 \cdot \boldsymbol \theta} {(1-p)^{1/\alpha} }\le \overline{\VaR}_p (\mathbf P_{\alpha,  \boldsymbol \theta} )   \le \frac{\alpha} {\alpha-1} \times \frac{ \mathbf 1 \cdot \boldsymbol \theta} {(1-p)^{1/\alpha} } .
\end{equation}
\end{enumerate}
\end{proposition}

\begin{proof}
  \begin{enumerate}[(i)]
    \item
    As the quantile of Pareto distribution is continuous, by Lemma 4.4 and 4.5 of \cite{BJW14}, $\overline{\VaR}_p(\mathbf P_{\alpha,\boldsymbol \theta} )$ is continuous in $p$ on  $(0,1)$.
    \item
    Let $\mathcal{U}$ be the set of uniform random variables distributed on $(0,1)$. We note that
    \begin{align*}
    \overline{\VaR}_p (\mathbf P_{\alpha, \boldsymbol \theta} )=&\sup\left\{\VaR_{p}\left(\theta_{1}U_{1}^{-1/\alpha}+\dots+\theta_{n}U_{n}^{-1/\alpha}\right):U_{1},\dots,U_{n}\in\mathcal{U}\right\}\\
    =&\sum_{i=1}^{n}\theta_{i}\sup\left\{\VaR_{1-p}\left(M_{\alpha, \boldsymbol \theta}(U_{1},\dots,U_{n})\right):U_{1},\dots,U_{n}\in\mathcal{U}\right\}^{-\frac{1}{\alpha}},
    \end{align*}
    where $M_{\alpha, \boldsymbol \theta}(u_{1},\dots,u_{n})=\left(\theta_{1}u_{1}^{-1/\alpha}+\dots+\theta_{n}u_{n}^{-1/\alpha}\right)^{-\alpha}/\left(\sum_{i=1}^{n}\theta_{i}\right)^{-\alpha}$, $u_{i}\in (0,1)$ for $i=1,\dots,n.$ Let $\underline{\boldsymbol \theta}=\min(\boldsymbol \theta/\left(\sum_{i=1}^{n}\theta_{i}\right)).$ With the classic averaging inequalities, for $0<\alpha_{1}<\alpha_{2}$, $M_{\alpha_{1}, \boldsymbol \theta}\le M_{\alpha_{2}, \boldsymbol \theta}$ (\cite{HLP34}, Theorem 16) and $\underline{\boldsymbol \theta}^{\alpha_{1}}M_{\alpha_{1}, \boldsymbol \theta}\ge \underline{\boldsymbol \theta}^{\alpha_{2}}M_{\alpha_{2}, \boldsymbol \theta}$ (\cite{HLP34}, Theorem 23). We note that $0<M_{\alpha, \boldsymbol \theta}<1$ and these two inequalities are directly translated to $$\overline{\VaR}_p (\mathbf P_{\alpha_{2}, \boldsymbol \theta} )^{\alpha_{2}/\alpha_{1}}\le \overline{\VaR}_p (\mathbf P_{\alpha_{1}, \boldsymbol \theta} )\le \underline{\boldsymbol \theta}^{1-\alpha_{2}/\alpha_{1}}\overline{\VaR}_p (\mathbf P_{\alpha_{2}, \boldsymbol \theta} )^{\alpha_{2}/\alpha_{1}}.$$
    By letting $\alpha_{1}\uparrow\alpha_{2}$ and $\alpha_{2}\downarrow\alpha_{1}$, we get the continuity of  $\overline{\VaR}_p (\mathbf P_{\alpha, \boldsymbol \theta} )$ in $\alpha>0$.

    \item

    Without loss of generality, we assume $\boldsymbol \theta_{1}=(\theta_{1},\dots,\theta_{n})$ and $\boldsymbol \theta_{2}=(\lambda\theta_{1},\dots,\theta_{n})$, $\lambda>0$. The monotonicity relative to $\boldsymbol \theta$ follows directly from Proposition \ref{prop:trivial4}.  Using the homogeneity of $\overline{\VaR}_p (\mathbf P_{\alpha, \boldsymbol \theta_{}} )$, which is proved in (iv), and the monotonicity with respect to $\boldsymbol \theta$ if $0<\lambda<1$,
    \begin{align*}
    \lambda\overline{\VaR}_p (\mathbf P_{\alpha, \boldsymbol \theta_{1}} )\leq \overline{\VaR}_p (\mathbf P_{\alpha, \boldsymbol \theta_{2}} )\leq\overline{\VaR}_p (\mathbf P_{\alpha, \boldsymbol \theta_{1}} ),
    \end{align*}
    otherwise
    \begin{align*}
    \overline{\VaR}_p (\mathbf P_{\alpha, \boldsymbol \theta_{1}} )\leq \overline{\VaR}_p (\mathbf P_{\alpha, \boldsymbol \theta_{2}} )\leq\lambda\overline{\VaR}_p (\mathbf P_{\alpha, \boldsymbol \theta_{1}} ).
    \end{align*}
    By letting $\lambda\uparrow 1$ and $\lambda\downarrow 1$, we get the desired result.

    \item
    For $\lambda>0$,
    \begin{align*}
    \overline{\VaR}_p (\mathbf P_{\alpha,\lambda \boldsymbol \theta} )&=\sup \{\VaR_p (G): G\in \mathcal D_n(\mathbf P_{\alpha,\lambda \boldsymbol \theta} )\}\\
    &=\sup \left\{\VaR_p \left(G\left(\frac{\cdot}{\lambda}\right)\right): G\in \mathcal D_n(\mathbf P_{\alpha, \boldsymbol \theta} )\right\}\\
    &=\lambda\sup \{\VaR_p ( G): G\in \mathcal D_n(\mathbf P_{\alpha, \boldsymbol \theta} )\}  =\lambda \overline{\VaR}_p (\mathbf P_{\alpha,\boldsymbol \theta} ).
    \end{align*}
    \item
    For $\alpha>1$, $\overline{\VaR}_p (\mathbf P_{\alpha,  \boldsymbol \theta} )\le \overline{\ES}_p (\mathbf P_{\alpha,  \boldsymbol \theta} )=\sum_{i=1}^{n}\ES_{p}(P_{\alpha,\theta_{i}})=\alpha \sum_{i=1}^{n}\theta_{i}/\left((\alpha -1 )   (1-p)^{1/\alpha} \right) ,$ and $\overline{\VaR}_p (\mathbf P_{\alpha,  \boldsymbol \theta} )\ge \sum_{i=1}^{n}\VaR_{p}(P_{\alpha,\theta_{i}})=\sum_{i=1}^{n}\theta_{i}/(1-p)^{1/\alpha}.$ \qedhere
  \end{enumerate}
\end{proof}

\end{document}